\pdfoutput=1
\documentclass[letterpaper,11pt]{article}
\usepackage[margin=1in]{geometry}
\usepackage{amsmath,amsthm,amssymb,amsfonts}
\usepackage[dvipdfmx]{graphicx}
\usepackage{wrapfig}
\usepackage{color}
\usepackage{svg}
\usepackage{url}
\usepackage{etoolbox}
\usepackage{bbm}
\usepackage{algorithm}
\usepackage{algorithmic}
\usepackage{xspace}
\usepackage[normalem]{ulem}
\usepackage{todonotes}

\usepackage{tikz}
\usetikzlibrary{positioning}
\usetikzlibrary{calc}

\tikzset{every picture/.style={font issue=\footnotesize},
            font issue/.style={execute at begin picture={#1\selectfont}}
           }

\usepackage[colorlinks=true,citecolor=blue,linkcolor=blue]{hyperref}
\usepackage[capitalise]{cleveref}

\newtheorem{theorem}{Theorem}[section]
\newtheorem{lemma}[theorem]{Lemma}
\newtheorem{corollary}[theorem]{Corollary}
\newtheorem{definition}[theorem]{Definition}
\newtheorem{proposition}[theorem]{Proposition}
\newtheorem{fact}[theorem]{Fact}
\newtheorem{folklore}[theorem]{Folklore}

\theoremstyle{remark}
\newtheorem{remark}[theorem]{Remark}

\newif\ifdraft
\drafttrue

\newcommand\modify[2]{
	\ifdraft
		\unskip\textcolor{red}{\sout{#1}}\textcolor{cyan}{#2}\unskip
	\else
		\unskip#2\unskip
	\fi
}

\def\defeq{\mathrel{\mathop:}=}
\newcommand{\numberhom}[1]{\#\textsc{Hom}^{(#1)}}
\newcommand{\numberemb}[1]{\#\textsc{Emb}^{(#1)}}
\newcommand{\numberembcol}[1]{\#\textsc{Emb}^{(#1)}_{\mathrm{col}}}

\newcommand{\kparity}[1]{\Pi^{\oplus #1}}
\newcommand{\numberembcolparity}[1]{\oplus \textsc{Emb}^{(#1)}_{\mathrm{col}}}

\newcommand{\distributionparity}[1]{\biguplus^k\distribution{#1}}

\newcommand{\existembcol}[1]{\textsc{Emb}^{(#1)}_{\mathrm{col}}}

\renewcommand{\hom}[2]{\numberhom{#1}(#2)}
\newcommand{\emb}[2]{\numberemb{#1}(#2)}
\newcommand{\embcol}[2]{\numberembcol{#1}(#2)}
\newcommand{\distribution}[1]{\mathcal{G}^{(#1)}_{n,1/2}}
\newcommand{\randombigraphdist}[2]{\mathcal{B}_{#1,#2,1/2}}

\newcommand{\DistKab}[3]{\mathcal{K}_{#1, #2, #3}}
\newcommand{\DistER}{\mathcal{G}_{n,1/2}}
\newcommand{\DistERDisjointUnion}{\biguplus^k\DistER}

\newcommand{\Unif}[1]{\mathrm{Unif}(#1)}
\newcommand{\dtv}{d_{\mathrm{TV}}}

\newcommand{\KabCount}{\numberemb{K_{a,b}}\xspace}
\newcommand{\ColKabDetect}{\textsc{Colorful $K_{a,b}$-Detection}\xspace}

\newcommand{\ColKaaDetect}{\textsc{Colorful $K_{a,a}$-Detection}\xspace}
\newcommand{\KaDetect}{\textsc{$K_{a}$-Detection}\xspace}
\newcommand{\KabParity}{\oplus K_{a,b}\mathchar`-\textsc{Subgraph}\xspace}
\newcommand{\KaParity}{\oplus K_{a}\mathchar`-\textsc{Subgraph}\xspace}

\newcommand{\embcolpoly}[3]{\mathsf{EMBCOL}_{#1,#2,#3}}
\newcommand{\distU}[3]{\mathcal{U}_{#1}^{(#2)}(#3)}
\newcommand{\IP}{\mathsf{IP}}

\DeclareMathOperator{\E}{\mathbf{E}}
\newcommand{\Var}{\mathop{\mathbf{Var}}}

\newcommand{\Fq}{\mathbb{F}_q}
\newcommand{\Nat}{\mathbb{N}}
\newcommand{\binset}{\{0,1\}}
\newcommand{\intset}[1]{\{0, \ldots,  #1\}}
\newcommand{\numset}[1]{\{1, \ldots,  #1\}}

\DeclareMathOperator\polylog{polylog}

\newcommand{\supp}{\mathsf{supp}\xspace}

\newcommand{\kOV}{$k$-\textsc{OV}\xspace}

\crefname{figure}{Figure}{Figures}
\crefname{remark}{Remark}{Remarks}
\crefname{equation}{}{}

\makeatother

\title{Nearly Optimal Average-Case Complexity of Counting Bicliques Under SETH} 
\author{
Shuichi Hirahara%
\thanks{
National Institute of Informatics.
Email: \texttt{s\_hirahara@nii.ac.jp}
}
\ and
Nobutaka Shimizu%
\thanks{The University of Tokyo.
Email: \texttt{nobutaka\_shimizu@mist.i.u-tokyo.ac.jp}}
}

\begin{document}
\maketitle

\begin{abstract}
In this paper, we seek a natural problem and a natural distribution of instances 
such that any $O(n^{c-\epsilon})$ time algorithm fails to solve most instances drawn from the distribution, while the problem admits an $n^{c+o(1)}$-time algorithm that correctly solves all instances.
Specifically,
we consider the $K_{a,b}$ counting problem in a random bipartite graph,
where $K_{a,b}$ is a complete bipartite graph and $a$ and $b$ are constants.
Our distribution consists of the binomial random bipartite graphs $B_{\alpha n, \beta n}$ with edge density $1/2$, where $\alpha$ and $\beta$ are drawn uniformly at random from $\{1,\ldots,a\}$ and $\{1,\ldots,b\}$, respectively.
We determine the nearly optimal average-case complexity of this counting problem
by proving the following results.
\begin{description}
  \item[Conditional Tight Worst-Case Complexity.] Under the Strong Exponential Time Hypothesis, for any constants $a\geq 3$ and $\epsilon>0$, there exists a constant $b=b(a,\epsilon)$ such that no $O(n^{a-\epsilon})$-time algorithm counts the number of $K_{a,b}$ subgraphs in a given $n$-vertex graph.
  On the other hand, for any constant $a\geq 8$ and any $b=b(n)$, we can count all $K_{a,b}$ subgraphs in time $bn^{a+o(1)}$.
  \item[Worst-to-Average Reduction.] If there exists a $T(n)$-time randomized heuristic algorithm that solves the $K_{a,b}$ subgraph counting problem on a random graph $B_{\alpha n,\beta n}$ with success probability $1-1/\mathrm{polylog}(n)$, then there exists a $T(n) \mathrm{polylog}(n)$-time randomized algorithm that solves the $K_{a,b}$ subgraph counting problem for any input with success probability $2/3$.
  \item[Fine-Grained Hardness Amplification.]
  Suppose that there is a $T(n)$-time algorithm with success probability $n^{-\epsilon}$ that
  computes the parity of the number of $K_{a,b}$ subgraphs in $H$, where
  $H:=G_1\uplus \cdots \uplus G_k$
  is the disjoint union of
  $k=O(\epsilon \log n)$ i.i.d.~random graphs $G_1,\ldots,G_k$ each of which is drawn from the distribution of $B_{\alpha n, \beta n}$.
  Then there is a $T(n) n^{O(\epsilon)}$-time randomized algorithm that counts $K_{a,b}$ subgraphs for any input with success probability $2/3$.
  \end{description}
  
The central idea behind these results is \emph{colorful subgraphs}.
For the first result,
we reduce the $k$-Orthogonal Vectors problem to the colorful $K_{a,b}$ detection problem.
In the second result,
we establish a worst-case-to-average-case reduction for a colorful subgraph counting problem
based on the binary-extension technique given by [Boix-Adser\`a, Brennan, and Bresler; FOCS19].
Then, we reduce colorful $K_{a,b}$ counting to $K_{a,b}$ counting.
Regarding the third result,
we prove the classical XOR lemma and the direct product theorem in the fine-grained setting for subgraph counting problems.
The core of the proof is an $O(\log n)$-round doubly-efficient interactive proof system for the colorful subgraph counting problem such that the honest prover is asked to solve $\mathrm{polylog}(n)$ instances of the counting problem.
The new protocol improves the known interactive proof system for the $t$-clique counting problem given by [Goldreich and Rothblum; FOCS18] in terms of query complexity.
\end{abstract}

\thispagestyle{empty}
\newpage
\tableofcontents
\thispagestyle{empty}
\newpage
\setcounter{page}{1}

\section{Introduction}
Understanding the average-case complexity
of a computational problem
is a fundamental question in the theory of computation.
One main reason is that
average-case hardness of a problem can serve as 
the first step towards building secure cryptographic primitives.
Besides,
average-case complexity reflects
the actual performance 
of 
an algorithm in the real world.

Motivated by a practical performance analysis, the average-case
complexity on a \emph{natural}
distribution has gathered special
attention.
A particular focus has been on
the complexity of problems on random graphs.
It is known that
several graph problems
such as 
\textsc{Hamiltonian Cycle} and
\textsc{Graph Isomorphism}, which are 
believed not to be in $\mathsf{P}$,
can be solved with high probability
in polynomial time
if the input is a random graph~\cite{FM97}.
Even for problems in $\mathsf{P}$, similar gaps between average- and worst-case complexity
have been observed.
For example, finding a maximum matching in an unweighted $m$-edge $n$-vertex graph admits an $O(m\sqrt{n})$-time algorithm~\cite{MV80}, while it admits an $\widetilde{O}(m)$-time algorithm that works on random graphs with high probability~\cite{Mot94}.%
\footnote{
    The notation $\widetilde{O}(\cdot)$ hides a $\polylog(n)$ factor.
}
Karp~\cite{Kar76} proposed the
problem of finding a clique
in a random graph:
He conjectured that
one cannot find a clique
of size $(1+\epsilon)\log n$ 
in an $n$-vertex Erd\H{o}s-R\'enyi
random graph of density $1/2$
in polynomial-time for any constant $\epsilon>0$.

The study of average-case complexity  has an application in \emph{Proof of Work} (PoW).
A PoW system, introduced by Dwork and Naor~\cite{DN92}, 
is a proof system that checks whether an untrusted party has consumed some amount of computational resources.
Such a proof system is useful for coping with a denial-of-service attack;
more recently it has been playing an important role in decentralized cryptocurrencies such as Bitcoin.
However, the criticism has been made that
the current Bitcoin system consumes considerable amounts of the computational resources of the miners
for what is a meaningless task (i.e., computing the SHA-256 hash).
To cope with this issue, 
Ball, Rosen, Sabin, and Vasudevan~\cite{BRSV17-2}
suggested a framework
in which the task of a prover is some meaningful one such as machine learning and finding primes.%
\footnote{ The original paper \cite{BRSV17-2} introduced the notion of \emph{Proof of Useful Work}.  However,
it turned out that a na\"ive protocol satisfies their definition \cite{BRSV18}, and it remains an open question to formulate the ideas of proof of useful work in a theoretically meaningful way.}
However, 
there are few natural problems whose average-case complexity is well understood,
which makes it difficult to construct a PoW system based on meaningful computational tasks.%
\footnote{ One approach to cope with this issue is to reduce meaningful computational tasks to an artificial problem whose average-case complexity is well understood, as suggested in \cite{BRSV17-2}; however, such a reduction is not necessarily efficient.}




The aim of this paper is to find a natural computational task whose worst- and average-case complexity has a ``sharp threshold''.
Specifically, for a threshold parameter $c$, we seek a natural problem and a natural distribution of instances 
such that any $n^{c-o(1)}$ time algorithm fails to solve \emph{most instances} drawn from the distribution, while the problem admits an $n^{c+o(1)}$-time algorithm that correctly solves all instances.
A problem which exhibits a sharp threshold between worst- and average-case complexity is required to construct a PoW.%
\footnote{ Roughly speaking, a PoW system requires an honest prover to solve the problem by using a worst-case solver in time $n^{c+o(1)}$, while a cheating prover running in time $n^{c-o(1)}$ fails to solve the problem on most instances.}
The existence of such an \emph{artificial} problem 
can be shown 
by using 
an average-case version of the time hierarchy theorem
(without relying on any unproven assumption) \cite{Wilber83_focs_conf,GoldreichW00_icalp_conf}.
However, prior to our work, no \emph{natural} problem whose sharp threshold between worst-case complexity and average-case complexity on a \emph{natural} distribution is determined under a widely investigated hypothesis.


As a natural computational task, 
we consider the problem of counting the number of bicliques (also known as complete bipartite graphs) on a natural distribution.
To state the problem more formally,
for a fixed graph $H$,
let $\numberemb{H}$ denote the problem that
asks the number of embeddings of $H$ in $G$ for a given graph $G$.
(An \emph{embedding} of $H$ in $G$ is an injective homomorphism from $H$ to $G$; see \cref{sec:preliminaries} for the definition of homomorphisms.)
The problem $\numberemb{H}$ is equivalent to the $H$-subgraph counting problem:
The number of $H$-subgraphs in $G$ is equal to the number of embeddings of $H$ in $G$ divided by the number of automorphisms of $H$.
Our main interest is the complexity of $\numberemb{K_{a,b}}$,
where $a$ and $b$ are constants
and $K_{a,b}$ is the
complete bipartite graph with
$a$ left vertices and $b$ right vertices.

The topic of finding or counting bicliques has been investigated from both practical and theoretical motivations.
On the practical side,
this study 
has applications in
 data mining~\cite{AS94,MT17} and bioinformatics~\cite{DABMMS04}.
See~\cite{MT17,AVJ98} and the references therein for details and lists of further applications.
On a theoretical side,
biclique detection and counting have been extensively studied 
especially in parameterized complexity theory  \cite{Lin18,CK12,GKL12,BFGL10,Kut12}.



For any constant $a \ge 8$,
our results determine a sharp threshold between worst-case complexity and average-case complexity
of counting the number of bicliques $K_{a, b}$ on a natural distribution for a large constant $b$,
under the Strong Exponential Time Hypothesis (SETH) of Impagliazzo, Paturi and Zane~\cite{IRPZ01}.
We consider the following distribution.

\begin{definition}[Random Bipartite Graph $\DistKab a b n$]
For given parameters $a,b,n\in\Nat$,
choose $\alpha, \beta$ uniformly at random from $\numset{a}$ and $\numset{b}$, respectively.
Let $\DistKab a b n$ be the distribution
of a random bipartite graph
with $n \alpha$ left vertices and $n \beta$ right vertices,
where each possible edge is included independently with probability $1 / 2$.
\end{definition}

Our first result determines a (not sharp) threshold between worst- and average-case complexity of  $(\numberemb{K_{a,b}}, \DistKab a b n)$ under SETH.
 
\begin{theorem}
[Worst- and Average-Case Complexity of Counting $K_{a, b}$]
\label{thm_Main}
The following holds.
\begin{itemize}
\item  For any constant $a \ge 8$ and
  any constant $b$,
  there exists a worst-case solver running in time $n^{a + o(1)}$
  that solves $\numberemb{K_{a,b}}$.
\item Under SETH
  \footnote{
    Strictly speaking, we assume a variant of SETH that rules out $(2-\epsilon)^n$-time \emph{randomized} algorithms solving $k$-SAT for some constant $k\geq k(\epsilon)$.
  },
  for any constants $\epsilon > 0$ and $a \ge 3$,
  there exists a constant $b=b(a,\epsilon)$ such that
  no average-case solver running in time $n^{a - \epsilon}$
  solves $\numberemb{K_{a,b}}$
  on a graph drawn from $\DistKab a b n$
  with success probability greater than $1-(1/\log n)^C$,
  where $C=C(a,b,\epsilon)$ is a sufficiently large constant.
  \end{itemize}
\end{theorem}

\Cref{thm_Main} is the first result 
that determines 
the nearly optimal average-case complexity
of \emph{natural distributional problems}
under the widely investigated hypothesis.
This result provides  insight towards understanding the hardest instance of subgraph counting problems.

However, there still remain two issues.
The first issue is that 
the success probability of the average-case solver mentioned in \cref{thm_Main} is required to be high.
As a consequence, the threshold between worst- and average-case complexity is not sharp enough.
The second issue is that the subgraph $K_{a,b}$ is not fixed in \cref{thm_Main}:
The parameter $b$ depends on $\epsilon$.
The problem of counting fixed $K_{a,b}$ might be more natural
in the context of complexity theory.

To deal with the first issue, we amplify the average-case hardness by considering the problem of 
solving multiple instances
of the $K_{a,b}$ subgraph counting problem.
Let $(\DistKab a b n)^k$
denote the distribution of
$k$ random graphs drawn independently from $\DistKab a b n$.
The following result 
exhibits a \emph{sharp} threshold between worst- and average-case complexity of 
the ``$k$-wise direct product'' of the distributional problem $(\numberemb{K_{a,b}}, \DistKab a b n)$.
\begin{theorem}[Average-Case Complexity of Counting $K_{a,b}$ for Multiple Instances]
\label{thm:fine-grained_direct_product_theorem_Kab}
Under SETH,
  for any constants $\epsilon > 0$ and $a \ge 3$,
  there exists a constant $b=b(a,\epsilon)$ such that
  no average-case solver running in time $n^{a - O(\epsilon)}$
  can
  count
  the numbers of $K_{a, b}$ subgraphs 
  in all the $k$ graphs drawn from the direct product $(\DistKab a b n)^k$
  with success probability greater than $n^{-\epsilon}$,
  where $k = O(\epsilon \log n)$.
\end{theorem}

Our proof techniques of amplifying average-case hardness can be applied to 
other subgraph counting problems.
Consider the problem 
$\KaParity$ of asking the parity of the number of $K_a$ subgraphs contained in a given input graph,
where $K_a$ denotes the clique of size $a$.
Let $\DistER$ denote the distribution of the Erd\H{o}s-R\'enyi random graph $G(n,1/2)$,
and
let $\DistERDisjointUnion$ denote the distribution of the disjoint union of $k$ random graphs $G_1,\ldots,G_k$ each of which is independently drawn from $\DistER$. 
We show that the distribution $\DistERDisjointUnion$
is a ``hardest'' distribution for $\KaParity$,
by presenting 
an error-tolerant worst-case-to-average-case reduction from $\KaParity$ to
the distributional problem $(\KaParity,\DistERDisjointUnion)$.

\begin{theorem}[Worst-Case to Average-Case Reduction for $\KaParity$] \label{thm:worst_to_average_KaParity-intro}
Let $\epsilon>0$ and $a \in \Nat$ be arbitrary constants.
If there exists a $T(n)$-time randomized heuristic algorithm that solves
$\KaParity$ on a graph drawn from $\DistERDisjointUnion$ with success probability greater than
$\frac{1}{2} + n^{-\epsilon}$ for any $k = O(\epsilon\log n)$,
then there exists a randomized algorithm
that solves $\KaParity$ on any input in time $T(n) n^{O(\epsilon)}$.
\end{theorem}
Since any decision problem can be solved with success probability $\frac{1}{2}$ by outputting a random bit,
the success probability of an average-case solver in \cref{thm:worst_to_average_KaParity-intro} is nearly optimal.
Therefore, \cref{thm:worst_to_average_KaParity-intro} shows that the decision problem $\KaParity$ exhibits \emph{some} sharp threshold between worst- and average-case complexity.%
\footnote{
The current fastest algorithm \cite{NP85} of counting $K_a$ subgraphs runs in time $O(n^{\omega \lceil a/3\rceil})$ on $n$-vertex graphs, where $\omega$ denotes the matrix multiplication exponent.
However, the precise value of $\omega$ is not well understood.
}

Boix-Adser{\`a}, Brennan, and Bresler~\cite{BBB19} and Goldreich~\cite{Gol20} 
presented worst-case-to-average-case reductions from $\KaParity$ to $(\KaParity, \DistER)$.
Their reductions are not error-tolerant and require the success probability of an average-case solver to be close to $1$.
They left an open question of improving the error tolerance of the reductions. \cref{thm:worst_to_average_KaParity-intro} improves the error tolerance, albeit for a slightly different distribution.

At the heart of the proof of
\cref{thm:fine-grained_direct_product_theorem_Kab,thm:worst_to_average_KaParity-intro}
is a doubly-efficient interactive proof system with subpolynomial number of queries,
which is of independent interest.

\begin{theorem}[Interactive Proof System for $K_{a,b}$ Counting with Subpolynomial Queries]
\label{thm:MainIP}
There is an $O(\log n)$-round interactive proof system
for $\numberemb{K_{a,b}}$
such that the verifier runs in time $O(n^2\log n)$
and asks the prover to
solve $\numberemb{K_{a,b}}$
for $\polylog n$ instances, where $n$ is the number of vertices of the given input graph.
\end{theorem}


To cope with the second issue of \cref{thm_Main} (i.e., the complexity of counting fixed $K_{a,b}$),
we consider the special case of $b=a$ and obtain the average-case hardness under Exponential Time Hypothesis (ETH).
In this paper, we consider a randomized variant of ETH which asserts that any $2^{o(n)}$-time randomized algorithm fails to solve $3$SAT.

\begin{theorem} \label{thm:Average_KaaCount_ETH}
Under ETH, any $n^{o(a)}$-time algorithm fails to count the number of $K_{a,a}$ subgraphs on more than a $(1-(1/\log n)^C)$ fraction of the inputs drawn from $\DistKab a a n$, where $C=C(a)$ is a constant that depends on $a$.
\end{theorem}

\section{Overview of Our Proof Techniques}

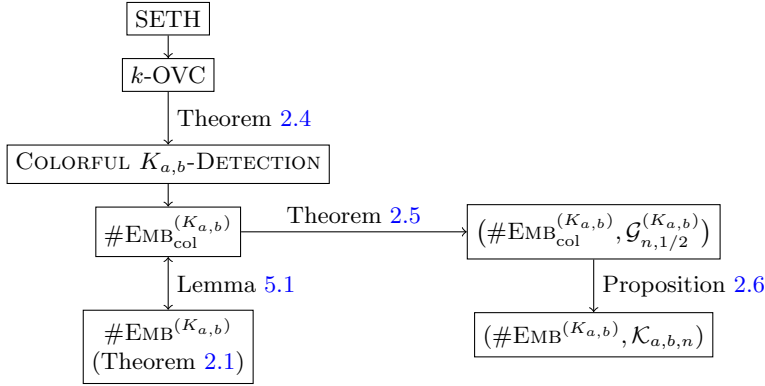
\begin{figure}[ht]
\begin{center}
\begin{tikzpicture}[
           normal node/.style={
               draw,
               text height=0.35cm,
               text centered
			},
           prop node/.style={
               draw,
               text height=0.2cm,
               text width=7em,               
               text centered
			},			
           tool node/.style={
               draw,
               text width=13em,
               text centered,
               text height=0.25cm,
               text centered
			}]
	\def\D{0.7};
	\def\S{0.3};
	\def\W{3};

  \node[draw] (SETH) at (0cm, 0cm) {SETH}; 
  \node[draw, align=center, below = \S of SETH] (kOVC) {$k$-OVC}; 
  \node[draw, align=center, below = \D of kOVC] (ColorfulKabDetection) {$\ColKabDetect$}; 
  \node[draw, align=center, below = \S of ColorfulKabDetection] (numEMBCOLKab) {$\numberembcol{K_{a,b}}$}; 
  \node[draw, align=center, below = \D of numEMBCOLKab] (KabCount) {$\KabCount$ \\ {\footnotesize (\cref{thm:KabCountSETH})}}; 
  \node[draw, align=center, right=\W of numEMBCOLKab] (numEMBCOLKab_avr) {$\bigl(\numberembcol{K_{a,b}}, \distribution{K_{a,b}} \bigr)$};

  \node[draw, align=center, below = \D of numEMBCOLKab_avr] (KabCount_avr) {$(\KabCount, \DistKab a b n)$}; 
  
  \node[above right = 0 and \W/2 of numEMBCOLKab.east, anchor=south] (worst_to_avr) {{\footnotesize \cref{thm:worst-to-average_precise}}}; 
  \node[below = \D/2 of numEMBCOLKab.south, anchor=west] (col_uncol) {{\footnotesize \cref{lem:coluncolequivalent}}}; 
  \node[below = \D/2 of kOVC.south, anchor=west] (kOVCtoColKabDetect) {{\footnotesize \cref{thm:mainthm1}}}; 
  \node[below = \D/2 of numEMBCOLKab_avr.south, anchor=west] (numEMBCOL_to_Kab_avr) {{\footnotesize \cref{prop:ToDistKab}}}; 

  \draw[->] (SETH) to (kOVC);
  \draw[->] (kOVC) to (ColorfulKabDetection);
  \draw[->] (ColorfulKabDetection) to (numEMBCOLKab);
  \draw[<->] (numEMBCOLKab) to (KabCount);
  \draw[->] (numEMBCOLKab) to (numEMBCOLKab_avr);
  \draw[->] (numEMBCOLKab_avr) to (KabCount_avr);

\end{tikzpicture}
\caption{The organization of the proof of average-cane hardness of $\KabCount$. \label{fig:organization}}
\end{center}
\end{figure}

In what follows, we briefly present the ingredients for our results while reviewing the literature.
The overall outline of the proof of \cref{thm_Main} is illustrated in \cref{fig:organization}.
In \cref{subsec:worst-case_subgraph_counting}, we describe the worst-case complexity of $\numberemb{K_{a,b}}$ under SETH and ETH.
In \cref{sec:average-case_complexity_of_subgraph_counting_problems}, we present our idea for the worst-case-to-average-case reduction for subgraph counting problems.
In \cref{sec:IP_intro}, we introduce our new doubly-efficient interactive protocols.
In \cref{sec:fine-grained_hardness_amplification}, we introduce a general framework of hardness amplification in fine-grained complexity settings.
We review related results in \cref{sec:previous}.





\subsection{Worst-Case Complexity of Subgraph Counting Problems} \label{subsec:worst-case_subgraph_counting}

\paragraph*{Nearly Optimal Complexity of $\#\mathsf{EMB}^{(K_{a,b})}$.}
Our first step is to determine the nearly optimal \emph{worst-case} complexity for $\KabCount$ under SETH
by proving the following results.

\begin{theorem} \label{thm:KabCountSETH}
For any constants $\epsilon>0$ and $a\geq 3$,
there exists a constant $b$ such that
one cannot solve $\KabCount$
in time $O(n^{a-\epsilon})$
unless SETH fails.
\end{theorem}

\begin{theorem} \label{thm:KaaCountETH}
If there exists an $n^{o(a)}$-time algorithm that solves $\numberemb{K_{a,a}}$,
then ETH fails.
\end{theorem}

\begin{proposition} \label{prop:fastalgo}
If $a\geq 8$, for any $\epsilon>0$ and $b\in\Nat$,
there is an algorithm
that solves
$\KabCount$ in time $O(b n^{a+\epsilon})$.
\end{proposition}

We are not aware of previous results that determine the nearly optimal complexity of subgraph counting problems,
while the fine-grained complexity of many natural problems, including the All-Pairs Shortest Paths, 3SUM, Orthogonal Vectors, and related problems, has been extensively explored in the research area of hardness in $\mathsf{P}$~\cite{Vir15,LPW17}.


\paragraph*{SETH-Hardness of $\#\mathsf{EMB}^{(K_{a,b})}$.}
Our key idea for showing
\cref{thm:KabCountSETH}
is to consider \ColKabDetect, which is defined as follows.
Let $K_n$ be the $n$-vertex complete graph.
For a graph $H$, let $K_n\times H$
denote the tensor product
(see \cref{sec:preliminaries} for the definition).
For a graph $G\subseteq K_n\times H$, every vertex $v=(u,i)\in V(G)$ is associated with a color $c(v)\defeq i\in V(H)$.
We say that a subgraph $F\subseteq K_n\times H$ is \emph{colorful} if $\bigcup_{v\in V(F)}\{c(v)\}=V(H)$.
The problem \ColKabDetect asks, given a pair $(n,G)$ of $n\in\Nat$ and a graph $G\subseteq K_n\times K_{a,b}$, to decide
whether $G$ contains a colorful subgraph
$F$ that is isomorphic to
$K_{a,b}$.
%

Exploiting the fact that \ColKabDetect is more ``structured'' than $\numberemb{K_{a,b}}$,
we first present a reduction from 
\textsc{$k$-Orthogonal Vectors} (\kOV) to
\ColKabDetect for $k := a$.
Since \kOV is known to be SETH-hard for any $k \ge 2$~\cite{Vir15,Wil05,LPW17}, this establishes SETH-hardness of  \ColKabDetect:
\begin{theorem} \label{thm:mainthm1}
For any constants $a\geq 2$ and $\epsilon>0$, there exists a constant $b=b(a,\epsilon)\geq a$ such that \textsc{Colorful $K_{a,b}$-Detection} cannot be solved in time $O(m^{a-\epsilon})$ unless SETH fails,
where $m$ is the number of edges of the input graph.
\end{theorem}

To complete the proof of \cref{thm:KabCountSETH},
we reduce \ColKabDetect to $\numberemb{K_{a,b}}$
by using 
the inclusion-exclusion principle.
This technique is well known 
in the literature of
fixed-parameter complexity (see, e.g.~\cite{CM14,Cur18}).
We will present the detail
in \cref{subsec:SETH-hardness_of_ColorfulKabDetect}.


\paragraph*{ETH-Hardness of $\#\mathsf{EMB}^{(K_{a,a})}$.}
The problem of finding a complete bipartite graph $K_{a,a}$ in a given graph has gathered special attention in parameterized complexity.
Lin~\cite{Lin15,Lin18} proved that the problem is W[1]-hard when $a$ is a parameter.
His proof implies that the problem of finding $K_{a,a}$ does not admit any $n^{o(\sqrt{a})}$-time algorithm unless ETH fails.
In particular, under ETH, any $n^{o(\sqrt{a})}$-time algorithm fails to solve $\numberemb{K_{a,a}}$. 

\cref{thm:KaaCountETH} improves this lower bound by ruling out an $n^{o(a)}$-time algorithm under ETH.
A key idea behind this improvement is to take advantage of the structure of \emph{counting} the number of embeddings:
We first reduce the problem of finding a clique $K_a$ of size $a$ (which is known to be ETH-hard \cite{CHKX06})
to  \textsc{Colorful $K_{a,a}$-Detection},
and then reduce it to $\numberemb{K_{a,a}}$ by using the inclusion-exclusion principle.
The latter reduction exploits the structure of counting.


\subsection{Average-Case Complexity of Subgraph Counting Problems} \label{sec:average-case_complexity_of_subgraph_counting_problems}
Compared to the worst-case hardness, the average-case hardness of subgraph counting problems has not been well understood until very recently~\cite{GR18,BBB19}.
A recent breakthrough result of 
Boix-Adser{\`a}, Brennan, and Bresler~\cite{BBB19}
shows that 
the worst-case and average-case complexities of 
counting $k$-cliques in an $n$-vertex Erd\H{o}s-R\'enyi graph
are equivalent up to a $\polylog(n)$-factor.
They left as an open question 
the extension of their results to other subgraph counting problems.

In this paper, we investigate their open question under a different setting, which is one of our key insights.
We consider the problem $\numberembcol{H}$ of counting $H$-\emph{colorful} subgraphs in a random graph drawn from a specific distribution $\distribution{H}$, which we introduce below.
%
For a fixed graph $H$,
let $\numberembcol{H}$ be the problem that
asks the number of colorful subgraphs that are isomorphic to $H$
contained in a given graph $G\subseteq K_n\times H$.
For $G\subseteq K_n\times H$, let $\embcol{H}{G}$ be the number of colorful $H$-subgraphs contained in $G$.
Equivalently,
$\embcol{H}{G}$ is equal
to the number of embeddings of $H$ in $G$
that preserves colors.
Here, we say that an embedding $\phi$ \emph{preserves} colors if $u=c(\phi(u))$ holds for any $u\in V(H)$, where $c:V(G)\to V(H)$ is the coloring of $G$.
See \cref{fig:embcol} for an illustration.
For a fixed graph $H$, let $G^{(H)}_{n,1/2}\subseteq K_n\times H$ be a random subgraph such that each edge $e\in E(K_n\times H)$ is included independently with probability $1/2$.
The distribution of $G^{(H)}_{n,1/2}$
is denoted by $\distribution{H}$.
We denote by
$(\numberembcol{H},\distribution{H})$
the distributional problem
of solving $\numberembcol{H}$
on a random graph 
drawn from $\distribution{H}$. 

\begin{figure}[htbp]
\centering
\includegraphics[width=9cm]{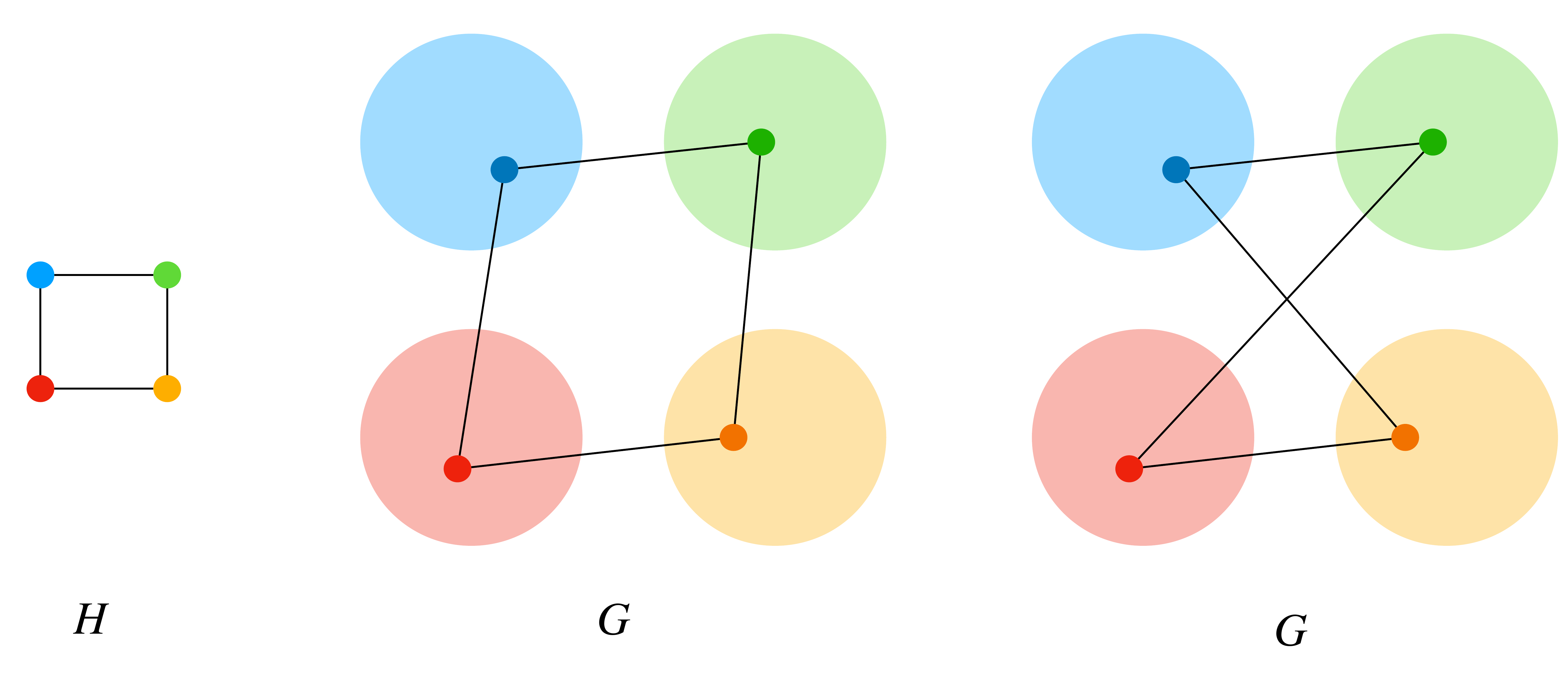}
\caption{An example of colorful subgraphs.
In this paper, we do not consider the case on the right-hand side (if $G\subseteq K_n\times H$, then $G$ contains neither blue-orange nor green-red edges).  \label{fig:embcol}}
\end{figure}

Generalizing the proof techniques of
Boix-Adser{\`a}, Brennan, and Bresler~\cite{BBB19},
we prove that
$\numberembcol{H}$ is reducible
to the distributional problem
$(\numberembcol{H},\distribution{H})$.

\begin{theorem}[Worst-Case-to-Average-Case Reduction for $\numberembcol{H}$]\label{thm:worst-to-average_precise}
Let $H$ be a fixed graph.
Suppose that the distributional problem $(\numberembcol{H},\distribution{H})$ can be solved by a $T(n)$-time randomized heuristic algorithm $A$ with success probability $1-\delta$, where $\delta=(\log n)^{-C}$ and $C=C_H$ is a sufficiently large constant depending on $H$.

Then, there is a $T(n)\cdot \polylog (n)$-time randomized algorithm $B$ that solves $\numberembcol{H}$ for any input with
success probability $2/3$.
Moreover, the number of oracle calls of $A$ by $B$ is at most $(\log n)^{O(|E(H)|)}$.
\end{theorem}

It should be noted that,
Dalirrooyfard, Lincoln, and Vassilevska Williams~\cite{DLW20} proved the same result as \cref{thm:worst-to-average_precise} (their work is independent to us).

Our technical contribution here is to reduce
$(\numberembcol{K_{a,b}},\distribution{K_{a,b}})$
to $(\KabCount,\DistKab a b n)$ using the inclusion-exclusion principle.

\begin{proposition}
\label{prop:ToDistKab}
Suppose that there is a $T(n)$-time randomized heuristic algorithm that solves $(\KabCount,\DistKab{a}{b}{n})$ with success probability $1-\delta$.
Then, there is an $O(ab2^{a+b}\cdot T(n))$-time randomized heuristic algorithm that solves $(\numberembcol{K_{a,b}}, \distribution{K_{a,b}})$ with success probability $1-O(ab2^{a+b}\delta)$.
\end{proposition}

Combining \cref{thm:worst-to-average_precise,prop:ToDistKab}, we obtain a worst-case-to-average-case reduction
from $\KabCount$ to $(\numberemb{K_{a,b}},\DistKab{a}{b}{n})$.

\begin{theorem} \label{thm:Kabworsttoaverage}
Let $2\leq a\leq b$ be arbitrary constants.
Suppose that there is a $T(n)$-time randomized heuristic algorithm that solves $(\numberemb{K_{a,b}},\DistKab{a}{b}{n})$ with success probability $1-\delta$, where $\delta=(\log n)^{-C}$ and $C=C(a,b)$ is a sufficiently large constant.

Then, there is a $T(n)\cdot \polylog ( n )$-time randomized algorithm that solves $\KabCount$ for any input with success probability $2/3$.
\end{theorem}

\Cref{thm:Kabworsttoaverage}
is of interest in its own right;
we emphasize that $a$ and $b$ can be chosen arbitrarily unlike \cref{thm:KabCountSETH} (i.e., the SETH-hardness of $\numberemb{K_{a,b}}$).
In the context of subgraph counting,
counting $K_{2,2}$ (i.e., $4$-cycle) subgraphs
in a graph on $n$ vertices with $m$ edges
attracts particular interest:
The current fastest counting algorithm
runs in time $O(n^{\omega})$ or $O(m^{1.48})$~\cite{AYZ97},
whereas finding a $K_{2,2}$ can be done
in time $O(n^2)$~\cite{YZ97} or $O(m^{1.41})$~\cite{AYZ97}.
A central question in this context is whether we can beat the $O(n^{\omega})$-time algorithm for the $K_{2,2}$-counting problem.
The worst-case-to-average-case reduction given in \cref{thm:Kabworsttoaverage}
indicates that a random bipartite graph
is essentially the hardest distribution for the $K_{2,2}$-counting problem.

\paragraph*{Proof Ideas.}
Our proof of \cref{thm:Kabworsttoaverage}
is based on the proof techniques of 
Boix-Adser{\`a}, Brennan, and Bresler~\cite{BBB19}.
They first reduced the clique counting problem to the clique counting problem on partite graphs.
Then they encoded the counting problem as a polynomial
over a large finite field and reduced the counting problem to the evaluation of the polynomial at a random point.
This is a classical method of
local decoding of
Reed-Muller codes,
which is a standard technique in the literature
of random self-reducibility~\cite{Lip91,GS92}.
The key part of the proof in \cite{BBB19} is a \emph{binary expansion technique}, which enables us to reduce
the polynomial evaluation to the clique counting problem on a random partite graph.
Specifically, they reduced the evaluation of $F(x)$ for a random point $x\in\Fq^N$
to the evaluation of $F(z_1),\ldots,F(z_m)$ for $m=\polylog(q)$, where $F(\cdot)$ is the polynomial that encodes the clique counting problem and $z_i\in \{0,1\}^N$ is a random binary point.
Finally, they reduced the clique counting problem on a random partite graph to the clique counting problem on an Erd\H{o}s-R\'enyi random graph by using an inductive argument.

The inductive argument given in \cite{BBB19} does not seem to be easily generalized to subgraphs other than cliques.
Instead, we present a simple proof that can be applied to bicliques (as well as cliques) by using the inclusion-exclusion principle,
and prove \cref{prop:ToDistKab}.

The proof of \cref{thm:worst-to-average_precise} is given by 
observing
that
the binary expansion technique of \cite{BBB19} 
can be applied to any multivariate polynomial $F$ over $\Fq$
on $N$ variables $x_1,\ldots,x_N$ satisfying the following conditions.
\begin{enumerate}
\item 
$F$ is a low-degree polynomial.
(The reduction runs in time 
$(\log n)^{O(D)}$, where $D$ is 
the total degree of $F$.)
\item There is a partition
$(E_1,\ldots,E_\ell)$ of $[N]$
such that $|E_j\cap I|=1$ holds for every monomial $\prod_{i\in I}x_i$ and $j\in[\ell]$.
\end{enumerate}
These two conditions motivate us
to consider 
the problem  $\numberembcol{H}$ of counting colorful subgraphs.
The recent work of Dalirrooyfard, Lincoln, and Vassilevska Williams~\cite{DLW20}
called the polynomial with this property \emph{strongly $\ell$-partite}
and presented a worst-case-to-average-case
reduction for evaluating this polynomial, which implies \cref{thm:worst-to-average_precise} (their work is independent to ours).
See \cref{subsec:step1_reduce_wrs_to_Fq} for details of the our worst-case-to-average-case reduction.

\subsection{Doubly-Efficient Interactive Proof System} \label{sec:IP_intro}
A line of research on interactive proof systems,
pioneered by Goldwasser, Micali, and Rackoff~\cite{GMR89},
revealed the surprising power of interaction. 
Early studies of interactive proof systems focused on efficient verification of intractable problems such as $\mathsf{PSPACE}$-complete problems~\cite{LFKN92,Sha90}.
A recent line of research (e.g., \cite{GKR15,RRR16,GR18-ITCS,GR18,BRSV18}) concerns interactive proof systems
for tractable problems, which are called \emph{doubly-efficient interactive proof systems}:
The goal of a doubly-efficient interactive proof system is to verify a statement in almost linear time%
\footnote{ We often use $n$ to denote the number of vertices of a given graph; thus, ``almost linear time (in the input length)'' means $\widetilde{O}(n^2)$ time in our context.}
by interacting with a polynomial-time prover.
It is worth mentioning that a doubly-efficient interactive proof system plays an important role in Proof of Work systems~\cite{BRSV17,BRSV18}.


\cref{thm:MainIP} provides a doubly-efficient interactive proof system for
the $K_{a,b}$ counting problem with $\polylog(n)$ queries.
More generally,
for any fixed graph $H$,
we present an interactive proof system
for the colorful $H$-subgraph counting problem $\numberembcol{H}$,
in which an honest prover is required to solve 
$(\numberembcol{H},\distribution{H})$ on average.

\begin{theorem}[IP for $\numberembcol{H}$] \label{thm:embcolIP}
Let $H$ be a graph.
There is an $O(\log n)$-round interactive proof system $\IP$ 
for the statement ``$\embcol{H}{G}=C$"
such that,
given an input $(G,n,C)$,
\begin{itemize}
    \item The verifier accepts with probability $1$ for some prover if the statement is true (perfect completeness), while it rejects for any prover
    with probability at least $2/3$ otherwise (soundness).
    \item In each round, the verifier runs in time $n^2(\log n)^{O(|E(H)|)}$ and 
    sends $(\log n)^{O(|E(H)|)}$ instances of $\numberembcol{H}$ to a prover.
\end{itemize}
Furthermore, for any constant $L_0$, there exists a constant $L_1=L_1(H,L_0)$ such that, 
if the statement is true and the prover has oracle
access to a randomized heuristic
algorithm that solves
$(\numberembcol{H},\distribution{H})$
with success probability
$1-(\log n)^{-L_1}$,
then the verifier accepts
with probability $1-(\log n)^{-L_0}$.
\end{theorem}

The ``Furthermore" part follows the worst-case-to-average-case reduction of 
\cref{thm:worst-to-average_precise}:
We can easily modify an honest prover of $\IP$
so that the prover
is required to solve
$\polylog(n)$ instances
of
the distributional problem $(\numberembcol{H},\distribution{H})$.

The salient feature of our interactive proof system is that 
the amount of communication between a verifier and a prover is at most $\polylog(n)$ bits;
equivalently, the number of queries that a verifier makes to a prover is at most $\polylog(n)$.
This will be important in the next section---where we prove hardness amplification theorems
in a fine-grained setting based on an interactive proof system whose query complexity is subpolynomially small.

The interactive proof system of \cref{thm:embcolIP} can be compared with one given by 
Goldreich and Rothblum~\cite{GR18}.
They presented an $O(1)$-round $\widetilde{O}(n)$-query doubly-efficient
interactive proof system for $\numberemb{K_k}$.
\cref{thm:embcolIP} significantly improves the query complexity from $\widetilde{O}(n)$ to $\polylog(n)$,
at the cost of increasing the round complexity from $O(1)$ to $O(\log n)$.
To explain the source of our improvement,
we review the ideas of \cite{GR18}: 
Their interactive proof system
is essentially a variant of the sum-check protocol~\cite{LFKN92}.
They encoded $\numberemb{K_k}$ as a polynomial over a large finite field
and used 
the following downward self-reducibility
of $\emb{K_k}{G}$:
$\emb{K_k}{G}=\sum_{i\in V(G)}\emb{K_{k-1}}{G-i}$,
where $G-i$ denotes the graph obtained
by removing the vertex $i$ from $G$.
In each round, the prover sends a polynomial of degree $O(n)$ to the verifier.
Each coefficient of the polynomial can be computed by calling a $\numberemb{K_{k}}$ solver $\polylog n$ times.
Overall, the number of queries made by the verifier 
is $O(n \polylog n)$.
To summarize,
the degree of the polynomial is the main bottleneck for the query complexity.


We improve the query complexity by exploiting a different type of downward self-reducibility.
Roughly speaking, at each round, we reduce verifying $\embcol{H}{G}$
for an $n$-vertex graph $G$
to the verification of $\embcol{H}{G_1},\ldots,\embcol{H}{G_m}$ for $m=\polylog(n)$, where each $G_i$ has $n/2$ vertices.
The downward self-reducibility enables us to encode the problem $\numberembcol{H}$ as a polynomial of degree $|E(H)|(2^{|V(H)|} - 1) = O(1)$ for a fixed graph $H$,
thereby reducing the query complexity.
The details are presented in \cref{sec:doubly_efficient_IP}.

We mention that the existence of a doubly-efficient interactive proof system 
with communication complexity $\polylog(n)$
for $\numberembcol{H}$
is guaranteed by using the general result of Goldwasser, Kalai, and Rothblum~\cite{GKR15}.
However,
the strategy of an honest prover of their proof system may not be computed efficiently with $\numberembcol{H}$ oracle.
We need an interactive proof system in which an honest prover is simulated with oracle access to $\numberembcol{H}$, as is guaranteed in \cref{thm:embcolIP}.
This will be important for the applications to hardness amplification theorems, as we explain next.

%
%

\subsection{Fine-Grained Hardness Amplification} \label{sec:fine-grained_hardness_amplification}

The error tolerance of 
the worst-case-to-average-case reduction from $\KabCount$ to the distributional problem $(\numberemb{K_{a,b}}, \DistKab{a}{b}{n})$
presented in \cref{thm:Kabworsttoaverage}
is not satisfactory: it requires a heuristic algorithm to solve random instances with probability at least $1 - 1 / \polylog(n)$.
This comes from the fact that we use a union bound 
in order to guarantee that all the $\polylog(n)$ queries are answered correctly,
exactly as in the work of Boix-Adser{\`a}, Brennan, and Bresler~\cite{BBB19}.
They left an open question of increasing the error tolerance of the reduction.
We answer the open question partially:
By using \emph{hardness amplification theorems}, we increase the error tolerance,
at the cost of modifying the distribution of a random graph (to some other natural distributions).

In this section,
we present a general framework for amplifying average-case hardness
in the fine-grained complexity settings,
based on the techniques from ``coarse-grained'' complexity theory.
Specifically, we prove fine-grained complexity versions of hardness amplification theorems
for any problem $f$ that admits an efficient \emph{selector} that makes $n^{o(1)}$ queries.
Such a selector for $f$ can be constructed from a doubly-efficient interactive proof system
in which (1) the verifier makes at most $n^{o(1)}$ queries and (2) the strategy of an honest prover can be efficiently computed given oracle access to $f$.
In particular, the interactive proof system of \cref{thm:embcolIP} enables us to establish
fine-grained complexity versions of hardness amplification theorems for $f = \numberembcol{H}$ for any fixed graph $H$.
We explain the details below.

\subsubsection{Direct Product Theorem}
A \emph{direct product theorem} is one of the fundamental hardness amplification results:
It states that,
if 
no small circuit can compute $f$ on more than a $(1 - \delta)$-fraction of inputs,
then
no small circuit can compute 
the $k$-wise direct product $f^k$
on a roughly $(1 - \delta)^k$-fraction of inputs.
Here, the $k$-wise direct product $f^k$ of $f$ is defined as $f^k(x_1, \dots, x_k) \defeq (f(x_1), \dots, f(x_k))$.
Our plan is to apply
a direct product theorem
for the function $f \defeq \numberemb{K_{a, b}}$
in order to amplify the average-case hardness of
the distributional problem $(\numberemb{K_{a, b}}, \DistKab{a}{b}{n})$.


However, there is an obstacle for applying hardness amplification to \emph{uniform computational models}
(as opposed to non-uniform computational models such as circuits).
A standard proof of a direct product theorem can be applied to only non-uniform computational models (cf.\ the survey of Goldreich, Nisan, and Wigderson~\cite{GoldreichNW11_goldreich2011_sp_books}).
Impagliazzo, Jaiswal, Kabanets, and Wigderson~\cite{ImpagliazzoJKW10_siamcomp_journals}
overcame this issue and
presented a direct product theorem that is applicable to
\emph{slightly non-uniform} algorithms.
Moreover,
their direct product theorem 
is highly optimized and simplified,
and thus it is 
applicable to the settings of fine-grained complexity.

We note that it is impossible to completely get rid of the non-uniformity from a direct product theorem
(if no property of a function $f$ being amplified is used%
\footnote{ As we will explain in \cref{subsubsection:selector_intro}, by exploiting a specific property of a function $f$ (i.e., the existence of a selector for $f$), we can obtain a completely uniform version of a direct product theorem.}).
In general,
a direct product theorem can be seen as an approximate version of a \emph{local-list-decoding} algorithm
of an error-correcting code. 
To be more specific, 
the function $f$ is encoded as the $k$-wise direct product $f^k$,
and the direct product theorem 
can be regarded as a local-list-decoding algorithm of $f^k$:
Given a circuit $C$ that solves $f^k$ on a roughly $\epsilon \approx (1 - \delta)^k$ fraction of inputs,
the local-list-decoding algorithm of \cite{ImpagliazzoJKW10_siamcomp_journals}
produces a list of candidate circuits $C_1, \cdots, C_m$,
one of which is guaranteed to solve $f$ on more than a $1 - \delta$ fraction of inputs,
where $m = O(1 / \epsilon)$.
The non-uniformity refers to the fact that $m \ge 2$. 
This is an inherent issue; in order to decode an error-correcting code 
from a highly noisy case (i.e., $\epsilon \approx 0$),
one cannot uniquely decode the error-correcting code.

The same issue arose in the work of Goldreich and Rothblum~\cite{GR18}.
In order to construct a strongly average-case hard distribution,
they used the local list-decoding algorithm of Sudan, Trevisan and Vadhan~\cite{SudanTV01_jcss_journals}.
In their case, the average-case distribution is constructed so that the problem of counting $k$-cliques corresponds to computing some low-degree polynomial.
Then, they used a low-degree tester and a self-corrector to obtain the correct value from the list of circuits.
However, we cannot exploit their technique of using the local list-decoding algorithm
since our goal is to obtain a \emph{natural} average-case hard distribution.
Instead, we invoke the direct product theorem of \cite{ImpagliazzoJKW10_siamcomp_journals}
and then use our doubly-efficient interactive
proof system (\cref{thm:MainIP}) to identify the correct circuit.

\subsubsection{Identifying a Correct Circuit by a Selector}\label{subsubsection:selector_intro}
In order to get rid of a small amount of non-uniformity,
we make use of a specific property of a function $f$.
The notion of (oracle) selector, introduced in \cite{Hirahara15_coco_conf},
exactly characterizes  the problem for which a small amount of non-uniformity can be removed (under any relativized world).
%
For problems $\Pi'$ and $\Pi$,
a \emph{selector from $\Pi'$ to $\Pi$}
is an efficient algorithm that solves the problem $\Pi'$
given oracle access to two oracles $A_0, A_1$ one of which is 
guaranteed to compute $\Pi$.
As shown in \cite{Hirahara15_coco_conf},
it is not hard to see that any selector that can identify a correct circuit among \emph{two} circuits can be modified to a selector
that can identify a correct circuit among \emph{many} circuits.
In light of this,
what is needed for applying
the direct product theorem of \cite{ImpagliazzoJKW10_siamcomp_journals}
is 
the existence of a selector from $\KabCount$ to the task of solving the distributional problem
$(\KabCount,\DistKab{a}{b}{n})$ with success probability $1 - \delta$.

In the settings of ``coarse-grained'' complexity \cite{Hirahara15_coco_conf},
it suffices to consider a polynomial-time selector
since polynomial-time algorithms can be composed nicely.
However, in the settings of fine-grained complexity,
one cannot afford even $n^{\Omega(1)}$ queries
for each candidate circuit,
because simulating the circuit takes time $n^{a - \epsilon}$.
The previous interactive proof system given by Goldreich and Rothblum~\cite{GR18} is not efficient enough 
in terms of the query complexity (cf.\ \cref{sec:IP_intro})

We overcome this difficulty by using the doubly-efficient interactive proof system
that makes at most $\polylog(n)$ queries (\cref{thm:MainIP}).
Roughly speaking,
we can construct a selector by simulating the verifier
of an interactive proof system
by using the candidate circuit as a prover.
More precisely, for a given input $x$
and two circuits $C_0$ and $C_1$,
the selector simulates $C_0$ and $C_1$ on input $x$
and obtains the two outputs $C_0(x)$ and $C_1(x)$.
Then, the selector runs the interactive proof
system to check whether the output is correct.
If one of $C_0$ or $C_1$ is correct,
the verifier accepts the corresponding output
and the selector outputs the accepted one.
In this way, by using \cref{thm:MainIP}, we construct a selector as stated in the following theorem.
\begin{theorem}[Selector for $\numberembcol{H}$ Using Subpolynomial Queries]
\label{thm:selector}
Let $C_1,\ldots,C_m$ be circuits such that, for some $i^*$, the circuit $C_{i^*}$ solves $(\numberembcol{H},\distribution{H})$ with success probability $1-(\log n)^{-K_H}$, where $K_H$ is a sufficiently large constant that depends only on $H$ and $m = \polylog(n)$.
Then, there is a randomized $n^2\polylog(n)$-time algorithm
that, given the list of circuits $C_1,\ldots,C_m$ as advice, 
solves $\numberembcol{H}$ correctly
with probability at least $2/3$
by making $\polylog(n)$ queries for each circuit $C_i$.
\end{theorem}

We emphasize the importance of low query complexity of a doubly-efficient interactive proof system.
Suppose that
we can simulate
the candidate circuits $C_0$ and $C_1$
in time $T_C(n)$
and that
the verifier runs in time $T_V(n)$,
making $Q(n)$ queries
in the interactive proof system.
The running time of a selector that is constructed from the interactive proof system
is roughly $O(T_V(n)+Q(n)T_C(n))$.
In our setting, $T_C(n)=n^{a-\epsilon}$ and
thus $Q(n)$ must satisfy $Q(n)=n^{o(1)}$.

Combining the ``almost uniform'' direct product theorem of \cite{ImpagliazzoJKW10_siamcomp_journals}
with the selector of \cref{thm:selector},
we obtain a \emph{completely uniform} and fine-grained version of a direct product theorem for the distributional problem $(\numberemb{K_{a,b}}, \DistKab{a}{b}{n})$,
which completes a proof of \cref{thm:fine-grained_direct_product_theorem_Kab}.

\subsubsection{Yao's XOR Lemma}
Let $f:\binset^n\to\binset$ be a Boolean function.
\emph{Yao's XOR lemma} asserts that,
if
no small circuit can compute $f$ on more than a $(1 - \delta)$ fraction of inputs,
then
no small circuit can compute 
$f^{\oplus k}$
on a roughly $\frac{1}{2} + (1 - \delta)^k$ fraction of inputs,
where $f^{\oplus k} \colon \binset^{n k} \to \binset$ is defined as $f^{\oplus k}(x_1, \dots, x_k) := f(x_1)\oplus \dots\oplus f(x_k)$.


An almost uniform version of Yao's XOR lemma
is given by
Impagliazzo, Jaiswal, Kabanets, and Wigderson~\cite{ImpagliazzoJKW10_siamcomp_journals}
by combining  their direct product theorem
with the local list decoding of the Hadamard code given by Goldreich and Levin~\cite{GL89}.
Since the local list decoding algorithm of~\cite{GL89}
is simple and efficient, we can apply it directly
to the fine-grained complexity.
As a consequence, we can prove a uniform and fine-grained version of Yao's XOR lemma for any problem that admits an efficient selector.

We apply the fine-grained version of Yao's XOR lemma
to the parity variant of $\numberembcol{H}$.
To state our results formally,
let $\numberembcolparity{H}$ denote the problem of computing the parity $\numberembcolparity{H}(G)\defeq (\numberembcol{H}(G)\mod 2)$ 
of the number of colorful embeddings of $H$ in a given graph $G$.
Observe that, for $k$ graphs $G_1,\ldots,G_k\subseteq K_n\times H$, computing $\numberembcolparity{H}(G_1)\oplus \cdots \oplus \numberembcolparity{H}(G_k)$ is equivalent to computing $\numberembcolparity{H}(G_1\uplus \cdots \uplus G_k)$, where $F\uplus G$ denotes the disjoint union of two graphs $F$ and $G$.
Let $\distributionparity{H}$ denote the distribution of $G_1\uplus\cdots\uplus G_k$, where each $G_i$ is independently chosen from $\distribution{H}$.

\begin{theorem}[XOR Lemma for $\numberembcolparity{H}$] \label{thm:XOR_for_embcolparity}
Let $H$ be an arbitrary graph and $c>0$ be an arbitrary constant.
Suppose that there is a $T(n)$-time randomized heuristic algorithm that solves
$(\numberembcolparity{H}, \distributionparity{H})$ for any $k=O(\log n)$ with success probability greater than $\frac{1}{2}+n^{-c}$.
Then, there exists a $T(n)n^{O(c)}$-time randomized algorithm that solves $\numberembcolparity{H}$ with probability at least $2/3$ on every input.
\end{theorem}

The proof of \cref{thm:XOR_for_embcolparity} is presented in \cref{sec:parity}.
The idea is to combine the fine-grained direct product theorem and the local list decoding of~\cite{GL89}.
Details can be found in \cref{sec:parity}.

 
\subsection{Related Work} \label{sec:previous}

\paragraph*{Complexity of Subgraph Counting.}
The problem $\numberemb{H}$ is a
fundamental task in the
context of graph algorithms.
For a general subgraph $H$,
we can solve $\numberemb{H}$
in time $f(k)\cdot n^{(0.174+o(1))\ell}$
for some function $f(\cdot)$,
where $k$ and $\ell$ are the number of vertices and edges of $H$, respectively~\cite{CDM17}.
If $H$ has some nice structural property (e.g., small treewidth), 
several faster algorithms are known (see \cite{Cur18} and the references therein).
However, to the best of our knowledge, there is no previous result that precisely determines
the complexity of counting subgraphs.
Chen, Huang, Kanj, and Xia~\cite{CHKX06} proved that one cannot find a $k$-clique in a given graph in time $f(k)\cdot n^{o(k)}$ for any function $f(\cdot)$ unless ETH fails.
The current fastest algorithm was given by Nes\v{e}t\v{r}il and Poljak~\cite{NP85}, who presented an $O(n^{\omega \lceil k/3\rceil})$-time algorithm that counts the number of $k$-cliques in a given $n$-vertex graph.
Here, $\omega<2.373$ is the square matrix multiplication exponent~\cite{LeGall14,Wil12}.
Lincoln, Vassilevska Williams, and Williams~\cite{LVW18}
imposed the assumption that detecting a $k$-clique
in an $n$-vertex graphs requires
time $n^{\omega k / 3-o(1)}$
and then derived a super-linear
lower bound for
the shortest cycle problem.
However, the precise value of $\omega$ is
currently not known, and, as a consequence, the precise time complexity of counting $k$-cliques
is not well understood.

\paragraph*{Complexity of Biclique Counting.}
We mention in passing some algorithmic results concerned with finding or counting bicliques.
The results below consider the case where $a$ and $b$ are given as input.
Binkele-Raible, Fernau, Gaspers, and Liedloff~\cite{BFGL10}
proved that, for given $a,b$ and a graph $G$,
one can find a $K_{a,b}$ subgraph in $G$ in time $O(1.6914^n)$.
Couturier and Kratsch~\cite{CK12} gave
an $O(1.6107^n)$-time algorithm 
for $\KabCount$.
They also provided an
$O(1.2691^n)$-time counting algorithm
that works on bipartite graphs.
It is known that
the number of
distinct maximal induced biclique
subgraphs
in any $n$-vertex graph is
$O(3^{n/3})=O(1.442^n)$~\cite{GKL12}.
If a given graph is bipartite,
one can solve $\KabCount$ by
enumerating all
maximal $K_{a,b}$ subgraphs
using a polynomial delay algorithm~\cite{MU04}.
Kutzkov~\cite{Kut12} presented
an $O(1.2491^n)$-time counting algorithm,
which is currently the fastest one.
If $a\leq b$ are small, we can solve $\KabCount$ in time $O(n^{a+1})$ by enumerating all size-$a$ vertex subsets.
If $a=2$, we can solve $\numberemb{K_{2,b}}$ in time $O(n^{\omega})$ by computing $A^2$, where $A\in\{0,1\}^{n\times n}$ is the adjacency matrix of a given graph.
To the best of our knowledge, no other algorithms are known for $\numberemb{K_{a,b}}$.

Finding $K_{a,a}$ is $\mathsf{NP}$-complete if $a$ is given as input~\cite{GJ79}.
The parameterized complexity of finding $K_{a,a}$ (parameterized by $k$) has gathered special attention.
Lin~\cite{Lin15,Lin18} proved the W[1]-hardness, which had been a fundamental open question.
Moreover, his proof implies that, assuming ETH, one cannot find a $K_{a,a}$ in time $n^{o(\sqrt{a})}$.
However, it still remains open whether ETH rules out an $n^{o(a)}$-time algorithm for 
the problem of finding a $K_{a,a}$ subgraph.
\cref{thm:KaaCountETH} rules out an $n^{o(a)}$-time algorithm for $\numberemb{K_{a,a}}$ under ETH.

\paragraph*{Average-Case Complexity in $\mathsf{P}$.}
A detection variant of $(\numberembcol{H},\distribution{H})$
(i.e., the problem of deciding whether
a given graph $G\subseteq K_n\times H$
contains $H$ as a 
colorful subgraph)
has been studied in the literature of
average-case circuit complexity
of the subgraph isomorphism problem (cf.\ Rossman~\cite{Ros18}).

In a pioneering work of Ball, Rosen, Sabin, and Vasudevan~\cite{BRSV17},
they initiated the study of average-case complexity in the context of fine-grained complexity.
Ball et al.~\cite{BRSV17} and
their subsequent work~\cite{BRSV18}
constructed average-case hard tasks
by encoding worst-case problems
by a low-degree polynomial
over a large finite field.
Based on techniques of random
self-reducibility (e.g.~\cite{CPS99}),
they explored the average-case
hardness of the
evaluation of this polynomial under
the worst-case assumptions including
the Orthogonal Vector Conjecture,
APSP Conjecture, and 3SUM Conjecture,
recent hot conjectures
in the study of hardness in $\mathsf{P}$~\cite{LPW17,Vir15}.
Their work is motivated by
the construction of PoW systems.
Due to the construction,
their average-case problems 
are artificial.

Goldreich and Rothblum~\cite{GR18} studied the average-case complexity of $\numberemb{K_k}$ for a constant $k$.
They presented a simple distribution over
$\widetilde{O}(n)$-vertex graphs on which
it is hard to count the number of
$k$-cliques with a
success
probability better than $3/4$.
The distribution is constructed by
a gadget reduction,
and it is somewhat artificial.
The key idea of their reduction is
to consider counting weighted cliques:
The input graph has node and edge weights in $\Fq$,
and the task is to compute the sum
of all weights of clique subgraphs.
The weight of a clique is defined
as the product of all node weights and edge weights contained in the clique.
They represented this counting problem
as a low-degree polynomial $P:\Fq^{n\times n}\to\Fq$ and
used polynomial interpolation
to reduce evaluating $P$ to computing $P(r)$, where $r\sim\Unif{\Fq^{n\times n}}$.
Combining the Chinese Reminder Theorem, a vertex-blowing-up technique and unifying multiple instances into one instance, they further
reduced evaluating $P(r)$ to solving $\numberemb{K_k}$ in a specific random graph.
Their result has an error tolerance of constant probability.
However, the blowing-up technique and unifying instances yielded an artificial random graph distribution.

The proof of \cref{thm:worst-to-average_precise}
is based on techniques of Boix-Adser{\`a}, Brennan, and Bresler~\cite{BBB19}, who reduced
$\numberemb{K_k}$
to $(\numberemb{K_k},\mathcal{G}_{n,p})$,
where $\mathcal{G}_{n,p}$ is the distribution
of an $n$-vertex Erd\H{o}s-R\'enyi graph with edge density $p$.
The reduction runs in time $p^{-1}n^2\polylog n$.
Here, the error probability
of the average-case
solver is assumed to be
at most $(\log n)^{-C}$
for a sufficiently large constant $C=C(k)$.
They also presented
a parity variant of $\numberemb{K_k}$
and obtained a worst-case-to-average-case reduction with a better error tolerance.

Independently of our work,
Dalirrooyfard, Lincoln, and Vassilevska Williams~\cite{DLW20}
reduced $\numberembcol{H}$ to $(\numberemb{H},\mathcal{G}_{n,p})$ for
a constant $p$.
They first reduced $\numberembcol{H}$ to
$(\numberembcol{H},\distribution{H})$
by the same way as the proof of \cref{thm:worst-to-average_precise}
and then reduced
$(\numberembcol{H},\distribution{H})$
to $(\numberembcol{H},\mathcal{G}_{n,p})$.
Using the techniques in the latter
reduction (called \emph{Inclusion-Edgeclusion} in their paper),
we can reduce
$(\numberembcol{H},\distribution{H})$
to
$(\numberemb{H},\DistER)$.


\paragraph*{Hardness Amplification in $\mathsf{P}$.}
The authors of \cite{GoldenbergS20_innovations_conf} studied hardness amplification of optimization problems,
including problems in P.
Unlike our settings (in which it is highly non-trivial to construct a selector as in \cref{thm:selector}),
it is trivial to construct a selector for any optimization problem;
therefore, it is easy to obtain hardness amplification theorems of optimizations problems
by using the powerful direct product theorem of 
Impagliazzo, Jaiswal, Kabanets, and Wigderson~\cite{ImpagliazzoJKW10_siamcomp_journals}.%
\footnote{We mention that the authors of \cite{GoldenbergS20_innovations_conf} do not seem to be aware of \cite{ImpagliazzoJKW10_siamcomp_journals}.} 

\subsection{Organization}
We present formal definitions
 of our framework
 in \cref{sec:preliminaries}.
In \cref{sec:average-case_complexity_of_embcol},
 we present the
 worst-case-to-average-case
 reduction for $\numberembcol{H}$
 and
 prove \cref{thm:worst-to-average_precise}.
In \cref{sec:complexity_of_Kab_counting},
 we investigate the
 worst-case complexity of
 $\numberemb{K_{a,b}}$.
In \cref{sec:doubly_efficient_IP},
 we present the doubly-efficient interactive proof system
 of \cref{thm:MainIP}.
In \cref{sec:fine-grained_hardness_amplification},
 we prove the direct product theorem in the setting of fine-grained complexity.
 Finally, in \cref{sec:parity},
 we prove our fine-grained XOR Lemma.

Here is the organization of the proofs of our main results.
\paragraph*{\cref{thm_Main}.}
The first statement follows from \cref{thm:KabCountSETH,thm:Kabworsttoaverage}.
We can obtain \cref{thm:Kabworsttoaverage} by combining \cref{thm:worst-to-average_precise,prop:ToDistKab}.
See \cref{sec:average-case_complexity_of_embcol} and \cref{subsec:step3} for the proofs of \Cref{thm:worst-to-average_precise} and \cref{prop:ToDistKab}, respectively.
The second statement is equivalent to \cref{prop:fastalgo}, which is shown in \cref{sec:fastalgo}.


\paragraph*{\cref{thm:fine-grained_direct_product_theorem_Kab}.}
The proof is given in \cref{subsection:DirectProductAndSelector}.

\paragraph*{\cref{thm:worst_to_average_KaParity-intro}.}
The proof is given in
\cref{subsec:KaParity}.

\paragraph*{\cref{thm:MainIP}.}
We obtain \cref{thm:MainIP} by combining \cref{thm:embcolIP} and \cref{lem:coluncolequivalent}.
\Cref{thm:embcolIP} and \cref{lem:coluncolequivalent} are shown in \cref{sec:doubly_efficient_IP} and \cref{sec:colvsuncol}, respectively.

\paragraph*{\cref{thm:Average_KaaCount_ETH}.}
This result follows from \cref{thm:KaaCountETH,thm:Kabworsttoaverage}.
\Cref{thm:KaaCountETH} is shown in \cref{sec:proof_of_Kab_results}.
\section{Formal Definitions} \label{sec:preliminaries}
\paragraph*{Notations and Computational Model.}
A graph is a pair $G=(V,E)$ of two finite sets $V$ and $E\subseteq \binom{V}{2}$.
For simplicity, we sometimes use $uv$ to abbreviate an edge $\{u,v\}$.
We denote by $V(G)$ and $E(G$) the vertex and edge set of $G$, respectively.
We sometimes identify a graph $G$ with a vector $x_G\in \binset^{E(H)}$ by regarding $x_G$ as the edge indicator of $G$.
For a finite set $S$ and a positive integer $k$, let $(S)_k=\{(s_1,\ldots,s_k):\{s_1,\ldots,s_k\}\in\binom{S}{k}\}$.

Throughout the paper, our computational model is the $O(\log n)$-Word RAM model.
As a consequence, we assume that any field operation can be done in constant-time if the underlying field is $\Fq$ with $q=n^{O(1)}$ ($n$ is specified by the problem).

\paragraph*{Subgraph Counting Problem.}
The \emph{tensor product} $X\times Y$ of two graphs $X$ and $Y$ is a graph given by $V(X\times Y)=V(X)\times V(Y)$ and $\{(x_1,y_1),(x_2,y_2)\}\in E(X\times Y)$ if and only if $\{x_1,x_2\}\in E(X)$ and $\{y_1,y_2\}\in E(Y)$.

For two graphs $G$ and $H$, a mapping $\phi:V(H)\to V(G)$ is \emph{homomorphism} from $H$ to $G$ if $\{\phi(u),\phi(v)\}\in E(G)$ whenever $\{u,v\}\in E(H)$.
An \emph{embedding} is an injective homomorphism.
Let $\emb{H}{G}$ be the number of embeddings from $H$ to $G$.

For a fixed graph $H$, we consider the problem $\numberemb{H}$ of computing $\emb{H}{G}$ for an input graph $G$.
Note that $\numberemb{H}$ is equivalent to the problem of counting subgraphs isomorphic to $H$:
If $\mathrm{Aut}(H)$ is the set of automorphisms
\footnote{An automorphism is a bijective homomorphism from $H$ to $H$.}
of $H$ and $\mathrm{sub}(H\to G)$
is the number of subgraphs of $G$ that is isomorphic to $H$, then $\emb{H}{G}=|\mathrm{Aut}(H)|\cdot \mathrm{sub}(H\to G)$.
In this paper, we consider
the colorful variant $\numberembcol{H}$
of $\numberemb{H}$, defined in \cref{sec:average-case_complexity_of_subgraph_counting_problems}.


\paragraph*{Distributional problems, Heuristics and Reduction.}
We regard a \emph{problem} $\Pi$ as a function from an input to the solution.
An algorithm is assumed to be randomized (unless otherwise stated).
An algorithm $A$ is said to solve a problem $\Pi$ in time $T(n)$
if $A$ runs in time $T(n)$ for any input $x$ of size $n$ and $\Pr_A [ A(x) = \Pi(x) ] \ge \frac{3}{4}$.

A \emph{distributional problem} is a pair $(\Pi,\mathcal{D})$ of a problem $\Pi$ and a family of distributions $\mathcal{D}=(\mathcal{D}_1,\mathcal{D}_2,\ldots)$, where each $\mathcal{D}_n$ denotes a distribution over inputs of size $n$.
To simplify notations, we shall refer to $(\Pi,D_n)$ rather than $(\Pi,(D_n)_{n\in\Nat})$.
Throughout the paper, $\Pi$ is a graph problem and each $\mathcal{D}_n$ is some random graph distribution.

In this paper, we follow the common notion of average-case complexity (e.g., \cite{BogdanovT06_fttcs_journals}).
We say that a (deterministic) heuristic algorithm $A$ solves a distributional problem $(\Pi, \mathcal{D})$ if,
for every $n \in \Nat$, 
$A$ outputs the solution of $\Pi$ on input $x$ with high probability over the random choice of $x \sim \mathcal{D}_n$.
The definition can be extended to a randomized (two-sided-error) heuristic algorithm:
\begin{definition}
[Randomized Heuristics \cite{BogdanovT06_fttcs_journals}]
\label{def:RandomizedHeuristics}
Let $(\Pi, \mathcal D)$ be a distributional problem and $\delta \colon \Nat \to [0, 1]$ be a function.
We say that a randomized algorithm $A$ \emph{solves} $(\Pi, \mathcal{D})$ with success probability $p$ if, for every $n \in \Nat$,
$\Pr_{x \sim \mathcal{D}_n} [ \Pr_A [ A(x; n) = \Pi(x) ] \ge \frac{3}{4} ] \ge p$.
Such an algorithm $A$ is called a \emph{(two-sided-error) randomized  heuristic algorithm} for $(\Pi, \mathcal{D})$.
\end{definition}

The reader is referred to the survey of Bogdanov and Trevisan \cite{BogdanovT06_fttcs_journals} for detailed background.

We will use the following notion of (fine-grained) reductions.
\begin{definition}
[Average-Case-to-Average-Case Reduction]
Let $(\Pi_1, \mathcal{D}_1), (\Pi_2, \mathcal{D}_2)$ be distributional problems.
Solving a distributional problem $(\Pi_1, \mathcal{D}_1)$ with success probability $1 - \delta_1$ is \emph{(quasi-linear-time polylog-query) reducible} to solving $(\Pi_2, \mathcal{D}_2)$ with success probability $1 - \delta_2$
if, there exists a $\widetilde{O}(n)$-time $\mathsf{polylog}(n)$-query randomized oracle algorithm $A$ such that
the algorithm $A$ solves
$(\Pi_1, \mathcal{D}_1)$ with success probability $1 - \delta_1$
given oracle access to an arbitrary algorithm that solves $(\Pi_2, \mathcal{D}_2)$ with success probability $1 - \delta_2$.
\end{definition}

\begin{definition}
[Worst-Case-to-Average-Case Reduction]
Let $\Pi_1$ be a problem, and let $(\Pi_2, \mathcal{D}_2)$ be a distributional problem.
A problem $\Pi_1$ is \emph{(quasi-linear-time polylog-query) reducible} to solving $(\Pi_2, \mathcal{D}_2)$ with success probability $1 - \delta_2$
if, there exists a $\widetilde{O}(n)$-time $\mathsf{polylog}(n)$-query randomized oracle algorithm $A$ such that
the algorithm $A$ solves
$\Pi_1$
given oracle access to an arbitrary algorithm that solves $(\Pi_2, \mathcal{D}_2)$ with success probability $1 - \delta_2$.
\end{definition}



%

\section{Average-Case Complexity of \texorpdfstring{$\#\mathsf{EMB}_{\mathrm{col}}^{(H)}$}{EMBcol}} 
\label{sec:average-case_complexity_of_embcol}

In this section, we present a proof of  \cref{thm:worst-to-average_precise}, that is, 
a worst-case-to-average-case reduction from
$\numberembcol{H}$
to the distributional problem
$(\numberembcol{H},\distribution{H})$.

For a fixed graph $H$ and a prime $q>n^{|V(H)|}$, let $\embcolpoly{n}{H}{q}:\Fq^{E(K_n\times H)} \to \Fq$ be a polynomial defined as
\begin{align}
\embcolpoly{n}{H}{q}(x) = \sum_{\substack{v_1,\ldots,v_k\in V(K_n\times H)\\ c(v_i)=i \ (\forall i)}} \prod_{\{i,j\}\in E(H)} x[v_iv_j]. \label{eq:embcolpolydef}
\end{align}
If $x\in\binset^{E(K_n\times H)}$ is the edge indicator of a graph $G\subseteq K_n\times H$, then $\embcolpoly{n}{H}{q}(x)=\embcol{H}{G} \bmod q = \embcol{H}{G}$ as $q>n^{|V(H)|}$.
For a graph $H$ and a set $S$, let $\distU{n}{H}{S}$ be the uniform distribution over $S^{E(K_n\times H)}$.
We sometimes identify $\Fq$ with the set $\intset{q-1}$.
The proof of \cref{thm:worst-to-average_precise}
consists of two steps.

\subsection{Step 1: Random Self-Reducibility of \texorpdfstring{$\embcolpoly{n}{H}{q}(\cdot)$}{Embcolpoly} over \texorpdfstring{$\distU n H \Fq$}{distU}}
\label{subsec:step1_reduce_wrs_to_Fq}
First, we reduce evaluating $\embcolpoly{n}{H}{q}(x)$ for a given $x$ to solving the distributional problem $(\embcolpoly{n}{H}{q}(\cdot), \distU{n}{H}{\Fq})$ for a large prime $q>n^{|V(H)|}$.
Note that we can obtain
such a prime $q$ as follows.
Sample a random integer $r$ from $\{n^{|V(H)|},n^{|V(H)|}+1,\ldots,2n^{|V(H)|}\}$
and then run the primality test for $r$ (according to the Prime Number Theorem, $r$ is prime with probability $\Omega(1/\log n)$.)

The following is well known in the context result of random self-reducibility.
A precise estimation of the running time was given by~\cite{BBB19,BRSV17}.

\begin{lemma}[Essentially given in Lemma 3.2 of~\cite{BRSV17}] \label{lem:localdecoding}
Let $P:\Fq^N \to \Fq$ be a multivariate polynomial of degree $d$ for a prime $q>12d$.
Suppose that there is a $T(N,q,d)$-time algorithm $A$ satisfying
\begin{align*}
\Pr_{x\sim \Unif{\Fq^N}}\left[A(x)=P(x)\right] \geq 1-\delta,
\end{align*}
where $\delta\in(0,1/3)$.
Then, there is a randomized algorithm $B$ that computes $P(y)$ on input $y\in\Fq^N$ with probability $2/3$ in time $O(Nd^2(\log q)^2+d^3+dT(N,q,d))$.
\end{lemma}
\begin{proof}[Proof sketch]
Ball et al.~\cite{BRSV17} proved this result under the condition that $d>9$.
Boix-Adser{\`a}, Brennan, and Bresler~\cite{BBB19} obtained the same result for a prime power $q>12d$ (under the same condition) by the same way.
The common idea is to invoke the well-known local decoding of the Reed-Muller code (see, e.g.,~\cite{Lip91,GS92}).
In this paper, we just modify a parameter appeared in their proof to remove the degree condition.
We briefly describe the algorithm and refer to the full version of~\cite{BRSV17} for the analysis.

For a given $y\in\Fq^N$, sample two random vectors $z_1,z_2\sim \Unif{\Fq^N}$ independently, and consider the univariate function $f(t):=y+z_1t+z_2t^2$.
Note that our task is to compute $f(0)$.
Set $m=100d$ (the authors of~\cite{BRSV17} set $m=12d$).
Use the oracle algorithm $A$ and compute $A(f(1)),\ldots,A(f(m))$.
By the Berlekamp--Welch decoding~\cite{BW86}, obtain a polynomial $\hat{f}$ and output $\hat{f}(0)$.
\end{proof}

By applying this result to our setting, we obtain the following.
\begin{corollary} \label{cor:worsttoaverage}
For a fixed graph $H$ and a prime $n^{|V(H)|}<q<2n^{|V(H)|}$, let $\embcolpoly{n}{H}{q}(\cdot)$ be the polynomial given in \cref{eq:embcolpolydef}.
Suppose that there is a $T(n)$-time algorithm $A$ satisfying
\begin{align*}
\Pr_{x\sim \distU{n}{H}{\Fq}}\left[A(x)=\embcolpoly{n}{H}{q}(x)\right] \geq \frac{2}{3}.
\end{align*}
Then, there is a randomized algorithm $B$ that computes $\embcolpoly{n}{H}{q}(y)$ on input $y\in\Fq^{E(K_n\times H)}$ with success probability $2/3$ in time $O(n^2(\log n)^2+T(n))$.
\end{corollary}

\subsection{Step 2: Reduce \texorpdfstring{$\embcolpoly{n}{H}{q}(\distU n H \Fq)$}{embcolpoly} to \texorpdfstring{$\embcolpoly{n}{H}{q}(\distU n H \binset)$}{embcolpoly}} \label{subsec:step2_reduce_Fq_to_01}

We reduce the problem of computing $\embcolpoly{n}{H}{q}(\cdot)$ over the distribution $\distU{n}{H}{\Fq}$ to that over $\distU{n}{H}{\binset}$ based on
the binary extension technique
of \cite{BBB19}.
Observe that the distributional problem $(\embcolpoly{n}{H}{q}(\cdot),\distU{n}{H}{\binset})$ is equivalent to $(\numberembcol{H},\distribution{H})$ if $q>n^{|V(H)|}$.

\begin{lemma} \label{lem:Fqtobin}
Let $H$ be a fixed graph and $q$ be a prime satisfying $n^{|V(H)|}<q<2n^{|V(H)|}$.
Suppose there is a $T(n)$-time randomized heuristic algorithm $A$ satisfying
\begin{align*}
    \Pr_{x\sim\distU{n}{H}{\binset)}}\left[\Pr_A\left[A(x)=\embcolpoly{n}{H}{q}(x)\right]\geq \frac{2}{3}\right] \geq 1-\delta,
\end{align*}
where $\delta=(\log n)^{-C}$ for a sufficiently large constant $C=C_H>0$ that depends on $H$.

Then, there is a $T(n)\cdot \polylog n $-time randomized heuristic algorithm $B$ satisfying
\begin{align*}
	\Pr_{x\sim\distU{n}{H}{\Fq}}\left[ \Pr_B\left[B(x)=\embcolpoly{n}{H}{q}(x)\right] > \frac{2}{3}\right] >\frac{2}{3}.
\end{align*}
\end{lemma}

Note that \cref{thm:worst-to-average_precise} follows from \cref{cor:worsttoaverage,lem:Fqtobin}.

\paragraph*{Observation.}
Suppose that, for each $uv\in E(K_n\times H)$, $x[uv]\in \Fq$ can be rewritten as 
\begin{align}
    x[uv]=\sum_{l=0}^{t-1} 2^l\cdot z^{(l)}[uv] \bmod q \label{eq:xbinaryexpansion}
\end{align}
for some binary variables $z^{(0)}[uv],\ldots,z^{(t-1)}[uv] \in \binset$.
Here, $t$ is some large integer that will be specified later.
Then, we obtain
\begin{align}
    \embcolpoly{n}{H}{q}(x)
    & = 
    \sum_{\substack{v_1, \ldots, v_k \in V(G) \\ c(v_i)=i \ (\forall i) }} \prod_{ij \in E(H)} \sum_{l=0}^{t-1} 2^l \cdot z^{(l)}[v_i v_j]
    \nonumber \\
    & = 
    \sum_{\substack{v_1, \ldots, v_k \in V(G) \\ c(v_i) =i \ (\forall i) }} 
    \sum_{a \in \intset{t-1}^{E(H)}}
    \prod_{ij \in E(H)}  2^{a[ij]} \cdot z^{(a[ij])}[v_i v_j]
        \nonumber \\
    & = 
    \sum_{a \in \intset{t-1}^{E(H)}}
    2^{\sum_{e \in E(H)} a[e]}
    \sum_{\substack{v_1, \ldots, v_k \in V(G) \\ c(v_i) \in i \ (\forall i) }} 
    \prod_{ij \in E(H)} z^{(a[ij])}[v_i v_j].
           \nonumber \\
    & = 
    \sum_{a \in \intset{t-1}^{E(H)}}
    2^{\sum_{e \in E(H)} a[e]}
    \cdot \embcolpoly{n}{H}{q}(\chi^{(a)}). \label{eq:EMBCOLpolybinaryexpand}
\end{align}
Here, we define $\chi^{(a)}[uv] \defeq z^{(a[c(u)c(v)])}[uv] \in \binset$ for each $uv\in E(K_n\times H)$.

Thus, our goal is to sample $z$ such that the distribution of $z^{(a)}$ is closed to $\distribution{H}$ for each $a\in\intset{t-1}^{E(H)}$.
In this paper, we invoke a special case of Lemma 4.3 of \cite{BBB19} and improve the running time of a sampling procedure.

\begin{lemma} \label{lem:binaryexpand}
Let $q>2$ be a prime and $t$ be some integer.
For each $x\in\Fq$, let $M_x\defeq \{m\in\intset{2^t-1}:m\bmod q=x\}$ and $Y_x\sim \Unif{M_x}$ be a random variable.
Let $\mathcal{Y}_R$ be the distribution of $Y_R$ for $R\sim\Unif{\Fq}$.
Then, the following hold. 
\begin{enumerate}
    \item $\dtv(\mathcal{Y}_R,\Unif{\intset{2^t-1}}) \leq Cq/2^t$ for some absolute constant $C$.
    \item For any given $x\in\Fq$, we can sample $Y_x$ in time $O(t)$.
\end{enumerate}
\end{lemma}
\begin{corollary} \label{cor:binarysampling}
Let $t$ be some integer.
Let $Z_0,\ldots,Z_{t-1}\sim \Unif{\binset}$ be i.i.d.~random variables.
Then, for any given $x\in\Fq$, we can sample $t$ random variables $z_0,\ldots,z_{t-1}$ satisfying the following in time $O(t)$.
\begin{enumerate}
    \item It holds that $\sum_{i=0}^{t-1}2^i\cdot z_i \bmod q = x$. \item The distribution of $(z_0,\ldots,z_{t-1})$ when $x$ is sampled from $\Unif{\Fq}$ is of total variation distance at most $O(q/2^t)$ from the uniform distribution $(Z_0,\ldots,Z_{t-1})$. 
\end{enumerate}
\end{corollary}
\begin{proof}
For a given $x\in\Fq$, let $z_0,\ldots,z_{t-1}$ be the binary expansion of $Y_x$ of \cref{lem:binaryexpand}.
Then, $Y_x=\sum_{i=0}^{t-1}2^i\cdot z_i =x \pmod q$ by the definition of $Y_x$.
Let $Y\defeq \sum_{i=0}^{t-1}2^i\cdot Z_i\sim \Unif{\intset{2^t-1}}$.
Let $f:\{0,\ldots,2^t-1\}\to\{0,1\}^t$ denote the function that maps $y\in\{0,\ldots,2^t-1\}$ to the binary representation of $y$.
Note that $f$ is a bijection and $f(Y_x)=(z_0,\ldots,z_{t-1})$ holds.
Then, from \cref{lem:binaryexpand}, for any $A\subseteq \binset^{t}$, we have
\begin{align*}
    |\Pr[(z_0,\ldots,z_{t-1})\in A]-\Pr[(Z_0,\ldots,Z_{t-1})\in A]| &= |\Pr[Y_x\in f^{-1}(A)]-\Pr[Y\in f^{-1}(A)]| \\ &= O(q/2^t).
\end{align*}
This implies the statement 2 of \cref{cor:binarysampling}.
\end{proof}
\begin{remark}
Boix-Adser{\`a}, Brennan, and Bresler~\cite{BBB19} considered the general case of $Z_i \sim \mathrm{Ber}(c_i)$, where $\mathrm{Ber}(c_i)$ is the Bernouli random variable with success probability $c_i$.
Roughly speaking, for some $t=\Theta(c^{-1}(1-c)^{-1}\log (q/\epsilon^2)\log q)$, they proved (1) $\dtv(\mathcal{L}(Y),\mathcal{L}(Y_R)) \leq \epsilon$, and (2) For any given $x\in\Fq$, $Y_x$ can be sampled in time $O(tq)$.
Since $q>n^{V(H)}$, the sampling of $Y_x$ cannot be applied directly due to the running time $O(tq)$.
To avoid the large running time, Boix-Adser{\`a}, Brennan, and Bresler~\cite{BBB19} used the Chinese Reminder Theorem to reduce computing $\embcolpoly{n}{H}{q}(\cdot)$ to the computing $\embcolpoly{n}{H}{q_1}(\cdot),\ldots,\embcolpoly{n}{H}{q_m}(\cdot)$, where $q_1,\ldots,q_m$ are small primes.
In \cref{lem:binaryexpand}, we focus on the special case of $c_i=1/2$ and improve the running time of sampling $Y_x$.
\end{remark}

We will present the proof of \cref{lem:binaryexpand}  later.

\paragraph*{Proof of \cref{lem:Fqtobin}.}
We describe the randomized algorithm $B$ that computes $\embcolpoly{n}{H}{q}(x)$ for a given $x\sim \distU{n}{H}{\Fq}$.

Set $t=K\log q$ for a sufficiently large constant $K=K(H)$ that will be chosen later depending only on $H$.
For each $e\in E(K_n\times H)$, do the following:
For $x=x[e]\in\Fq$, sample $z[e]=(z_0[e],\ldots,z_{t-1}[e])$
of \cref{cor:binarysampling} in time $O(t)$.
Note that \cref{eq:xbinaryexpansion} holds.

After sampling $(z[e])_{e\in E(K_n\times H)}$, the algorithm $B$ computes $\embcolpoly{n}{H}{q}(x)$ using \cref{eq:EMBCOLpolybinaryexpand}:
For each $a\in\intset{t-1}^{E(H)}$, construct $\chi^{(a)}$ using $(z[e])_{e\in E(K_n\times H)}$ and compute $\embcolpoly{n}{H}{q}(\chi^{(a)})$ using the $T(n,H)$-time heuristic algorithm $A$ that solves $(\embcolpoly{n}{H}{q}(\cdot),\distU{n}{H}{\binset})$ with success probability $1-\delta$.

We claim that $B$ has success probability $1-t^{|E(H)|}\delta - O(n^2|E(H)|q/2^t)$, which completes the proof of \cref{lem:Fqtobin}:
Indeed, choosing $t=K\log n$ for a sufficiently large constant $K=K(H)$, the success probability of $B$ is at least $1-O(\delta (\log n)^{2|E(H)|}) - o(1) \geq 2/3$ if $\delta = o((\log n)^{-2|E(H)|})$.

\paragraph*{Success probability of $B$.}
Since $x[e]\sim\Unif{\Fq}$, \cref{lem:binaryexpand} implies that the distribution of $z[e]\defeq (z_i[e])_{i\in\intset{t-1}}$ is total variation distance at most $\epsilon\defeq O(q/2^t)$ from that of $Z[e]\defeq (Z_0[e],\ldots,Z_{t-1}[e])$, where $Z_0[e],\ldots,Z_{t-1}[e]\sim\Unif{\binset}$ are i.i.d.~random variables.
Therefore, the distribution of $z=(z[e])_{e\in E(K_n\times H)}$ is total variation distance at most $|E(K_n\times H)|\epsilon$ from $Z=(Z[e])_{e\in E(K_n\times H)}$ (here, $z[e]$ are independent as well as $Z[e]$).

Let $A$ be the randomized heuristic algorithm described in \cref{lem:Fqtobin}.
Let $\mathcal{S}$ be the set of graphs that is solved by $A$.
Formally,
\begin{align*}
\mathcal{S}=\left\{F\subseteq K_n\times H:\Pr_{A}[A(F)=\embcol{H}{F}] \geq \frac{3}{4}\right\}.
\end{align*}

Let $z\defeq (z[e])_{e\in E(K_n\times H)}$ and $Z\defeq (Z[e])_{e\in E(K_n\times H)}$ be random variables described above.
For each $a\in\intset{t-1}^{E(H)}$, we have $\Pr_Z\left[\tilde{\chi}^{(a)} \in \mathcal{S}\right] \geq 1-\delta$, where $\tilde{\chi}^{(a)}=(\tilde{\chi}^{(a)}[e])_{e\in E(K_n\times H)}$ is defined as $\tilde{\chi}^{(a)}[uv]\defeq Z^{(a[c(u)c(v)])}[uv]$.
Here, we identify a graph with a binary vector in $\binset^{E(K_n\times H)}$.
Recall that $c:V(K_n\times H)\to V(H)$ maps a vertex to its color.
Note that the distribution of $\tilde{\chi}^{(a)}$ is the same as $\distribution{H}$ for every fixed $a\in\intset{t-1}^{E(H)}$.
By the union bound, we have
\begin{align*}
    \Pr_Z\left[\forall a\in\intset{t-1}^{E( H)}:\tilde{\chi}^{(a)} \in \mathcal{S}\right] \geq 1-t^{|E(H)|}\delta.
\end{align*}
Since $z$ is total variation distance at most $|E(K_n\times H)|\epsilon$ from $Z$, this implies
\begin{align*}
    \Pr_z\left[\forall a\in\intset{t-1}^{E(H)}:\chi^{(a)} \in \mathcal{S}\right] \geq 1-t^{|E(H)|}\delta-|E(K_n\times H)|\epsilon.    
\end{align*}
This completes the proof of the claim.

\paragraph*{Proof of \cref{lem:binaryexpand}.}
Indeed, the statement 1 is a special case of Lemma 4.3 in \cite{BBB19} and the proof is already given (see p.~23 of~\cite{BBB19}).
For completeness, we present the proof by focusing on the special case.
Consider the size of $M_x$.
Let $N\defeq 2^t/q$.
Since $x\in\{0,\ldots,q-1\}$, it holds that
\begin{align*}
     N-2\leq 
     \left\lfloor \frac{2^t}{q}-1\right\rfloor \leq |M_x| \leq  \left\lfloor \frac{2^t}{q}\right\rfloor
     \leq  N.
\end{align*}
Let $Y\sim \Unif{\intset{2^t-1}}$ and $Y_R\sim 
\mathcal{Y}_R$ be random variables, where $R\sim \Unif{\Fq}$.
For any $A\subseteq \intset{2^t-1}$, consider the events that $Y\in A$ and $Y_R \in A$.
Observe
\begin{align*}
    \Pr[Y_R\in A] = \sum_{x\in\Fq}\Pr[Y_x\in A\cap M_x|R=x]\Pr[R=x] 
    = \frac{1}{q}\sum_{x\in\Fq}\frac{|A\cap M_x|}{|M_x|}
\end{align*}
and
\begin{align*}
    \Pr[Y\in A] = \frac{|A|}{2^t} 
    = \frac{1}{q}\sum_{x\in \Fq} \frac{|A\cap M_x|}{N}.
\end{align*}
Therefore, it holds for any $A\subseteq\intset{2^t-1}$ that
\begin{align*}
    |\Pr[Y_R\in A] - \Pr[Y\in A]| &\leq \frac{1}{q}\sum_{x\in \Fq}|A\cap M_x|\left| |M_x|^{-1}-N^{-1} \right| \\
    &\leq \frac{|A|}{q}\left(\frac{1}{N-2}-\frac{1}{N}\right) \\
    &= \frac{|A|}{q}\cdot O(N^{-2}) \leq O(q/2^t).
\end{align*}
This completes the proof of the statement 1.

We show the statement 2.
The sampling can be done by the following scheme:
For a given $x\in \mathbb{F}_q$, let $M\defeq \lfloor (2^t-x-1)/q\rfloor=|M_x|-1$ and sample $K\sim\Unif{\intset{M}}$.
Then, output $L\defeq Kq+x$.
For any $k\in\{0,\ldots,M\}$,
\begin{align*}
    \Pr[L=kq+x] = \Pr[K=k] = \frac{1}{M+1}.
\end{align*}
In other words, $L\sim \Unif{M_x}$ for any $x$.

\section{Complexity of \texorpdfstring{$\#\mathsf{EMB}^{(K_{a,b})}$}{Counting Kab}}
\label{sec:complexity_of_Kab_counting}
This section is devoted to prove \cref{thm:KabCountSETH,prop:fastalgo,thm:Kabworsttoaverage,thm:KaaCountETH}.
In \cref{sec:ETH-hardness,sec:colvsuncol,subsec:SETH-hardness_of_ColorfulKabDetect,subsec:step3,sec:fastalgo}, we provide several technical results.
Finally, in \cref{sec:proof_of_Kab_results}, we combine these results to show \cref{thm:KabCountSETH,prop:fastalgo,thm:Kabworsttoaverage,thm:KaaCountETH}.

\subsection{Colorful Subgraph Counting vs.~(Uncolored) Subgraph Counting} \label{sec:colvsuncol}
We first prove that the $K_{a,b}$-subgraph counting and colorful $K_{a,b}$-subgraph counting are equivalent.
\begin{lemma} \label{lem:coluncolequivalent}
Consider $\numberembcol{K_{a,b}}$ and $\KabCount$.
Given oracle access to one of them, we can solve the other one in time $2^{O(a+b)}+O(n^2)$ (in the worst-case sense).
\end{lemma}

One direction is well known:
The problem
$\numberembcol{H}$
is reducible to
$\numberemb{H}$
by using the inclusion-exclusion principle~\cite{CM14,Cur18}.

\begin{folklore} \label{folk:uncolortocolor}
Let $H$ be a graph.
If $\numberemb{H}$ for $n$-vertex graphs can be solved in time $T(n)$, 
then
$\numberembcol{H}$ can be solved in time $O(2^{|V(H)|} T(n))$.
\end{folklore}

Now we discuss the converse direction:
Can we solve $\numberemb{H}$
given oracle access to $\numberembcol{H}$?
We show that $\numberemb{H}$ is reducible to $\numberembcol{H}$ when $H = K_{a, b}$.
To this end,
we consider the problem $\numberhom{H}$
that asks the number $\hom{H}{G}$ of homomorphisms from $H$ to a given graph $G$.
Recall that a mapping $\phi:V(H)\to V(G)$ is a homomorphism if $\{\phi(u),\phi(v)\}\in E(G)$ whenever $\{u,v\}\in E(H)$.

We reduce
$\numberemb{K_{a, b}}$ to $\numberembcol{K_{a, b}}$ by the following three steps.
First, we show that
$\hom{H}{G}$ is equal to $\numberembcol{H}(G \times H)$ (\cref{fact:homembcol}).
Second,
we use Lov\'asz's identity~\cite{Lov12}
to reduce $\numberemb{H}$ to $\numberhom{H'}$ for some family of graphs $H'$ (\cref{thm:lovaszinverse}).
Finally, we observe that $\numberhom{H'}$ is reducible to $\numberemb{K_{a, b}}$ when $H = K_{a, b}$ (\cref{prop:KcdpitoKab}).

The following well-known fact asserts
that $\numberhom{H}$ is reducible to $\numberembcol{H}$. 

\begin{fact} \label{fact:homembcol}
Let $H$ be a fixed $k$-vertex graph.
For any graph $G$, it holds that $\hom{H}{G}=\embcol{H}{G\times H}$.
Consequently, if $\numberembcol{H}$ can be solved in time $T(kn)$ on $kn$-vertex graphs, then $\numberhom{H}$ can be solved in time $O(T(kn)+kn^2)$ on $n$-vertex graphs.
\end{fact}
\begin{proof}
We can solve $\numberhom{H}$ on input $G$ as follows.
Construct $G\times H$ and then run the algorithm for $\numberembcol{H}$ on input $G\times H$.
Now we show $\hom{H}{G}=\embcol{H}{G\times H}$.
Let $\phi$ be a homomorphism from $H$ to $G$.
Then, the mapping $\psi:V(H)\ni v \mapsto (\phi(v),v)\in V(G\times H)$ is also a homomorphism and moreover it is injective.
This correspondence between $\phi$ and $\psi$ is one-to-one.
\end{proof}

In light of \cref{fact:homembcol}, 
it suffices to reduce $\numberemb{H}$ to
$\numberhom{H}$.
To this end, we invoke the following identity.

\begin{theorem}[Lov\'asz \cite{Lov12}; See (2) of \cite{CDM17}] \label{thm:lovaszinverse}
Let $H$ be a fixed graph.
Let $\mathcal{P}(H)$ be the set of partitions of $V(H)$ such that, for every $\pi=\{B_1,\ldots,B_t\}\in \mathcal{P}(H)$, each $B_i\subseteq V(H)$ is an independent set \upshape{(}$i=1,\ldots,t$\upshape{)}.
For each $\pi\in \mathcal{P}(H)$, define $H/\pi$ as the graph obtained by contracting each vertex set in $\pi$.
Then
\begin{align*}
    \emb{H}{G}=\sum_{\pi\in\mathcal{P}(H)} (-1)^{|V(H)|-|\pi|}\prod_{B\in\pi}(|B|-1)!\cdot \hom{H/\pi}{G}.
\end{align*}
Here, $|\pi|$ denotes the number of subsets in $\pi$.
\end{theorem}

Combining \cref{thm:lovaszinverse,fact:homembcol}, we can reduce $\numberemb{H}$ to solving
a family of problems $(\numberembcol{H/\pi})_{\pi\in\mathcal{P}(H)}$.
If $H=K_{a,b}$, we can enumerate all elements of $\mathcal{P}(H)$ in time $O(2^{a+b})$,
and thus the reduction runs in time $O(n^2 + 2^{a+b})$.
Moreover, we show in \cref{prop:KcdpitoKab} that 
$\numberembcol{K_{a,b}/\pi}$ is reducible to 
$\numberembcol{K_{a,b}}$
for 
every $\pi\in\mathcal{P}(K_{a,b})$,
which enables us to reduce $\numberemb{K_{a,b}}$ to $\numberembcol{K_{a,b}}$.

\begin{proposition} \label{prop:KcdpitoKab}
Assume that $\numberembcol{K_{a,b}}$ can be solved in time $T(n)$.
Let $\pi \in \mathcal{P}(K_{a,b})$.
Then, $\numberembcol{K_{a,b}/\pi}$ can be solved in time $O(n^2+ T(n))$.
\end{proposition}
\begin{proof}
Observe that, for any $\pi\in\mathcal{P}(K_{a,b})$, we have $K_{a,b}/\pi=K_{c,d}$ for some constants $c\leq a$ and $d\leq b$;
therefore, it suffices to reduce $\numberembcol{K_{c, d}}$ to $\numberembcol{K_{a,b}}$.

Let $(n,G)$ be an input of $\numberembcol{K_{c,d}}$, where $G\subseteq K_{c,d}\times K_n$.
Regard the vertices in $V(K_{a,b})$ as $V(K_{a,b})=\{l_1,\ldots,l_a,r_1,\ldots,r_b\}$ so that $E(K_{a,b})=\{l_i, r_j\}_{i\in[a],j\in[b]}$.
Then, each vertex $v\in V(G)$ can be represented as the form $(r_i,u)$ or $(l_i,u)$.
We write $V(G)=R\cup L$, where $R$ is the set of vertices of the form $(r_i,u)$, and $L$ is that of the form $(l_i,u)$.
Fix a vertex $v\in V(K_n)$ and let $L_{\mathrm{add}}=\{(l_i,v)\}_{i=c+1}^{a}$ and $R_{\mathrm{add}}=\{(r_i,v)\}_{i=d+1}^{b}$ be vertex sets.
We construct a graph $\hat{G}\subseteq K_{a,b}\times K_n$ as follows.
\begin{align*}
    V(\hat{G})&=V(G)\cup L_{\mathrm{add}} \cup R_{\mathrm{add}},\\
    E(\hat{G})&=E(G)\cup E(R_{\mathrm{add}},L\cup L_{\mathrm{add}}) \cup E(L_{\mathrm{add}},R\cup R_{\mathrm{add}}),
\end{align*}
where, for two vertex subsets $S$ and $T$, $E(S,T)=\{s\modify{}{, }t\}_{s\in S,t\in T}$.
See \cref{fig:Ghat} for an illustration.

\begin{figure}[htbp]
\centering
\includegraphics[width=5cm]{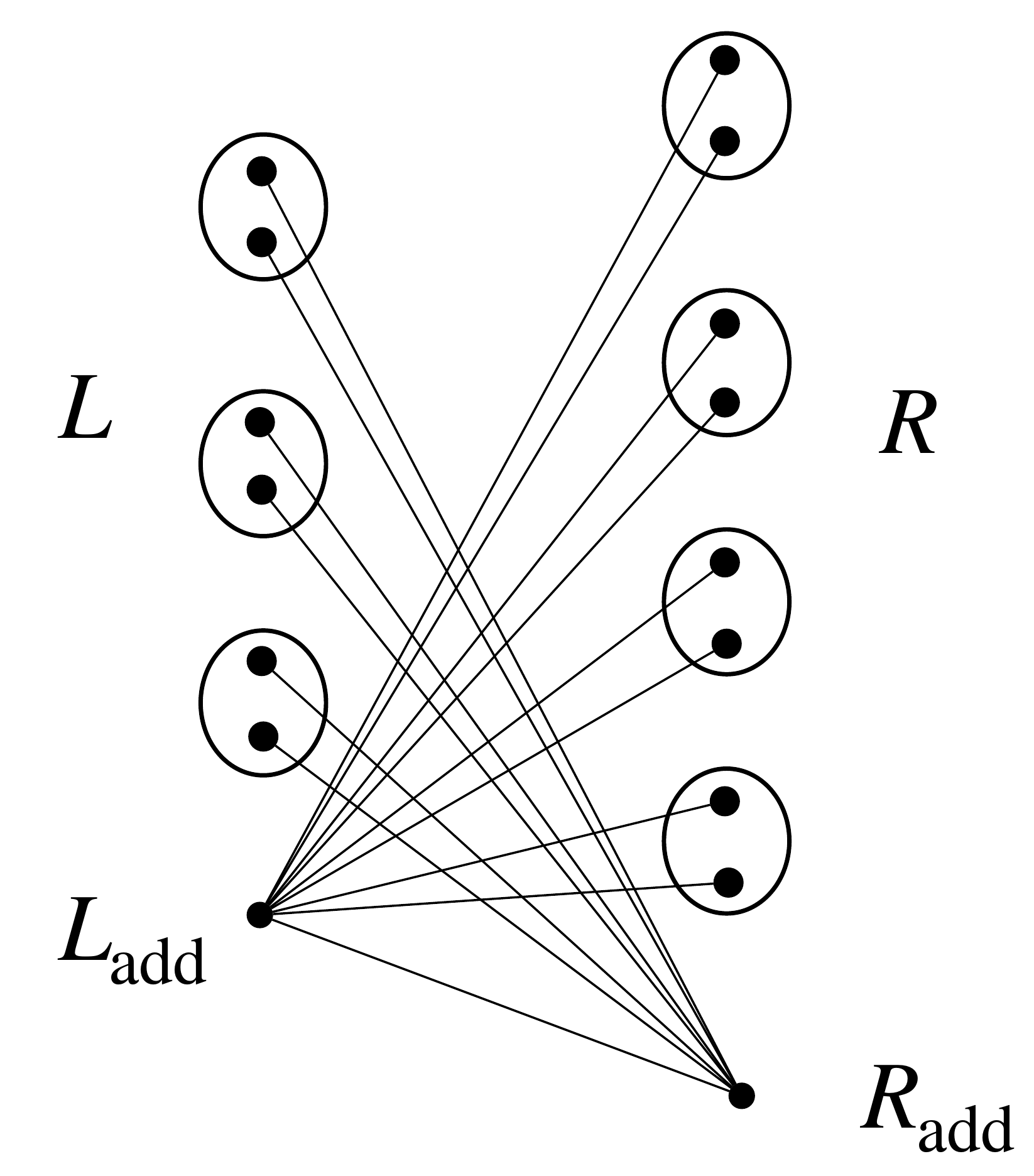}
\caption{The graph $\hat{G}$ of the reduction.
In this figure, $\embcol{K_{3,4}}{G}=\embcol{K_{4,5}}{\hat{G}}$ holds. \label{fig:Ghat}}
\end{figure}

Note that $\embcol{K_{a,b}}{\hat{G}}=\embcol{K_{c,d}}{G}$ holds since there is a one-to-one correspondence between copies of $K_{a,b}$ in $\hat{G}$ and that of $K_{c,d}$ in $G$.
\end{proof}

\Cref{lem:coluncolequivalent} follows from \cref{folk:uncolortocolor,fact:homembcol,thm:lovaszinverse,prop:KcdpitoKab}.

\begin{remark}
We comment on the relationship between $\numberhom{H}$ and $\numberemb{H}$.
It is easy to see that the problems $\numberhom{K_k}$ and $\numberemb{K_k}$ are equivalent.
More generally, such an equivalence holds 
if $H$ is a \emph{core}; here, a graph $H$ is said to be a core if any homomorphism from $H$ to $H$ is an isomorphism.
However, for some $H$, it is widely believed that there is a gap between $\numberhom{H}$ and $\numberemb{H}$:
For example, let $M_k$ be the graph of disjoint $k$ edges.
It is known that $\numberemb{M_k}$ (which is the problem of counting the number of matchings of size $k$) is $\#\mathrm{W}$[1]-hard~\cite{Cur13}, while $\numberhom{M_k}$ can be solved in linear time (observe that $\emb{M_k}{G}=(2|E(G)|)^k$).
\end{remark}

\subsection{SETH-Hardness of COLORFUL \texorpdfstring{$K_{a,b}$}{Kab}-DETECTION}
\label{subsec:SETH-hardness_of_ColorfulKabDetect}
Assume $a\leq b$.
By enumerating all subsets of size $a$, we can solve both \ColKabDetect in time $O(n^{a+1})$.
If a given graph $G$ is sparse and has $m$ edges, we can solve the problem in time $O(m^a)$ by enumerating $\binom{N(v)}{a}$ for every vertex $v$, where $N(v)$ denotes the set of vertices adjacent to $v$.

\begin{theorem}[Reminder of \cref{thm:mainthm1}]
For any constants $a\geq 2$ and $\epsilon>0$, there exists a constant $b=b(a,\epsilon)\geq a$ such that \textsc{Colorful $K_{a,b}$-Detection} cannot be solved in time $O(m^{a-\epsilon})$ unless SETH fails,
where $m$ is the number of edges of the input graph.
\end{theorem}

\begin{remark}
\Cref{thm:KabCountSETH} immediately follows from \cref{folk:uncolortocolor,thm:mainthm1}.
\end{remark}

In the proof of \cref{thm:mainthm1}, we consider \textsc{$k$-Orthogonal Vectors} (\kOV).
In \kOV, we are given sets
$A_1,\ldots,A_k\subseteq \{0,1\}^d$
of binary vectors
each of cardinarity $n$
and dimension $d$
satisfying $d\leq K\log n$
for a constant $K$.
Our task is to decide whether
there exist vectors
$\mathbf{a}_1,\ldots,\mathbf{a}_k$ such that
$\mathbf{a}_i\in A_i$ for any $i$ and
$\sum_{j=1}^d\prod_{i=1}^k \mathbf{a}_i[j]=0$.
The na\"ive exhaustive search solves $k$-OV in time $O(n^kd)=O(n^k\log n)$.
The current known fastest algorithm solves it in time $O(n^{k-1/O(\log(d/\log n))})$~\cite{AWY15}.
The \emph{$k$-Orthogonal Vectors Conjecture} ($k$-OVC) asserts that \kOV requires time
$n^{k-o(1)}$ for any $d=\omega(\log n)$:
More precisely, under $k$-OVC,
for any $k\geq 2$ and $\epsilon>0$, there exists a constant $K\geq 1$ such that no $O(n^{k-\epsilon})$-time algorithm solves \kOV of dimension $d\leq K\log n$.
It is known that, for every constant $k\geq 2$, SETH implies $k$-OVC~\cite{Vir15,Wil05,LPW17}.
Thus, it suffices to reduce \kOV to $\KabCount$ for $k=a$.

\paragraph*{The reduction (Proof of Theorem~\ref{thm:mainthm1}).}
Fix any constant $a\geq 2$.
Assume that there exists a constant $\epsilon>0$ such that \ColKabDetect can be solved in $O(m^{a-\epsilon})$ for every $b\geq a$.
We will prove that, under this assumption, there exists a constant $\epsilon'>0$ such that $a$-\textsc{OV} of dimension $d = K \log n$ can be solved in time $O(m^{a-\epsilon'})$ for any $K$.
To this end, we present a many-to-one reduction: The reduction maps an instance of $a$-\textsc{OV} to an equivalent instance of \ColKabDetect.

Let $\epsilon>0$ be a sufficiently small constant that will be specified later.
Let $A_1,\ldots,A_k\subseteq \{0,1\}^d$ be an instance of $k$-OV of dimension $d=K\log n$.
We identify a vector $\mathbf{x}\in\{0,1\}^d$ with a subset $x\subseteq [d]$.
Thus, each $A_i$ is identified with a family of subsets of $[d]$.
Let $\mathcal{P}_1\cup\cdots\cup\mathcal{P}_C$ be the partition of $[d]$ such that $|\mathcal{P}_i|\leq \epsilon \log n$ holds for every $i\in[C]$, where $C=K/\epsilon$ (we will choose $\epsilon$ so that $K/\epsilon$ is an integer).

The reduction constructs a graph $G$ and a coloring $c:V(G)\to [a+b]$, resulting in an instance of \ColKabDetect of $a:=k$ and $b:=C$.
The vertex set of $G$ is of the form
\begin{align*}
V(G)=V_1\cup\cdots\cup V_k\cup W_1\cup\cdots\cup W_C,
\end{align*}
where each subset $V_1, \ldots, V_k, W_1, \ldots, w_C$ is assigned with a distinct color.
For each subset $a\in A_i$, we create a vertex $v_{a}\in V_i$.
For each index $j\in[C]$, enumerate all subsets of $\mathcal{P}_j$.
We associate a $k$-tuple $z=(y_1,\ldots,y_k) \in (2^{\mathcal{P}_j})^k$ of the subsets with a vertex $w_z\in W_j$, if the corresponding vectors $\mathbf{y}_1,\ldots,\mathbf{y}_k$ are orthogonal on $\mathcal{P}_j$.
Formally, the vertex set $V(G)$ is
\begin{align*}
&V_i:=\{v_a:\,a\in A_i\},\\
&W_j:=\left\{w_z:\,z=(y_1,\ldots,y_k)\in (2^{\mathcal{P}_j})^k\text{ satisfies }\sum_{r\in\mathcal{P}_j}\prod_{s\in[k]} \mathbf{y}_{s}[r]=0 \right\}.
\end{align*}
Two vertices $v_a\in V_i$ and $w_z\in W_j$ of $z=(y_1,\ldots,y_k)$ are joined by an edge if $a\cap \mathcal{P}_j = y_i\cap \mathcal{P}_j$ holds.
The edge set $E(G)$ contains no other edges.
Note that $G\subseteq K_n\times K_{a,b}$ in which $V_1,\ldots,V_a,W_1,\ldots,W_b$ obtain distinct colors.

\paragraph*{Correctness.}
Let $A_1,\ldots,A_k\subseteq \binset^{[d]}$ be the instance of $k$-OV with $d=K\log n$ and $G$ be the graph constructed by the reduction above.
Recall that each $A_i$ is identified with a family of $n$ subsets of $[d]$.
Suppose that the given instance is a YES-instance.
Then, there is a $k$-tuple $(a_1,\ldots,a_k)\in A_1\times\cdots\times A_k$ such that the corresponding vectors $\mathbf{a}_1,\ldots,\mathbf{a}_k$ satisfy $\sum_{r\in[d]}\prod_{s\in[k]}\mathbf{a}_s[r]=0$.
Let
\begin{align*}
&U:=\bigcup_{i\in[k]}\{v_{a_i}\}\subseteq V(G),\\
&W:=\bigcup_{j\in[C]}\{w_z\in W_j:\text{$z=(y_1,\ldots,y_k)$ where each $y_i$ satisfies $y_i\cap \mathcal{P}_j = a_i\cap \mathcal{P}_j$}\}.
\end{align*}
The set $U\cup W$ induces a colorful subgraph isomorphic to $K_{a,b}$, where $a=k$ and $b=C$; thus, the pair $(G,c)$ of a graph $G$ and coloring $c$ is a YES-instance of \textsc{Colorful $K_{a,b}$-Detection}.

Conversely, suppose that $G$ contains a colorful subgraph $H$ isomorphic to $K_{a,b}$.
Then we have $|V(H)\cap V_i|=|V(H)\cap W_j|=1$ for every $i\in[k]$ and $j\in[C]$.
Let $v_i\in V(H)\cap V_i$ and $w_j\in V(H)\cap W_j$.
As $H$ is isomorphic to $K_{a,b}$, $\{v_i,w_j\}\in E(G)$ for every $i,j$.
Let $\mathbf{a}_i\in\{0,1\}^d$ be the vector associated with the vertex $v_i$.
For every $j\in[C]$, we have $\sum_{r\in\mathcal{P}_j}\prod_{s\in[k]}\mathbf{a}_s[r]=0$ since each $w_j$ is incident to $v_i$ for all $i\in[k]$.
Thus, we have $\sum_{r\in[n]}\prod_{s\in[k]}\mathbf{a}_s[r]=0$ and hence $(A_1,\ldots,A_k)$ is a YES-instance of $k$-OV.

\paragraph*{Running time.}
The size of the constructed graph $G$ satisfies
\begin{align*}
&|V(G)|\leq kn+Cn^{\epsilon k},\\
&|E(G)|\leq kCn^{1+\epsilon k}.
\end{align*}
Thus, if \ColKabDetect on $G$ can be solved in time $O(m^{a-\epsilon'})$, letting $\epsilon>0$ be a constant satisfying $(1+\epsilon k)(k-\epsilon')\leq k-\epsilon'/2$ yields an
\begin{align*}
O(m^{a-\epsilon'})=O(n^{(1+\epsilon k)(k-\epsilon')}) = O(n^{k-\epsilon'/2})
\end{align*}
time algorithm for \kOV.
This falsifies $k$-OVC as well as SETH.

\subsection{ETH-Hardness of COLORFUL \texorpdfstring{$K_{a,a}$}{Kaa}-DETECTION}
\label{sec:ETH-hardness}

Consider the decision problem \KaDetect in which we are asked to decide whether the given graph contains a clique of size $a$ or not.
In this section, we reduce \KaDetect to \ColKaaDetect.
Note that \KaDetect does not admit an $f(k)\cdot n^{o(k)}$-time algorithm for any function $f(\cdot)$ unless ETH fails~\cite{CHKX06}; thus, the reduction establishes the ETH-hardness of $\ColKaaDetect$.

\begin{lemma} \label{lem:Ka_to_Kaa_reduction}
There is an $O(n^2)$-time algorithm that, given a graph $G$ of $n$ vertices, outputs a graph $G'\subseteq K_n\times K_{a,a}$ of $O(an)$ vertices such that $G$ contains an $a$-clique if and only if $G'$ contains a colorful $K_{a,a}$-subgraph.
\end{lemma}
\begin{proof}
Let $G$ be an instance of \KaDetect.
We transform $G$ to the graph $G'$ mentioned in \cref{lem:Ka_to_Kaa_reduction}.

Let $U_1,\ldots,U_a,W_1,\ldots,W_a$ be copies of $V(G)$.
For notational convenience, we write $V(G)=\{v_1,\ldots,v_n\}$ and $U_i=\{u^{(i)}_1,\ldots,u^{(i)}_n\},\,W_i=\{w^{(i)}_1,\ldots,w^{(i)}_n\}$.
Here, each $u^{(i)}_j$ corresponds to $v_j$ (and so does $w^{(i)}_j$).
We set $V(G')=\bigcup_{i\in [a]} (U_i \cup W_i)$; each $V_i$ and $W_i$ is assigned with a distinct color. (More formally, each vertex in $V_i$ is assigned with a color $i$ and each vertex in $W_i$ is assigned with a color $a+i$).
We construct $E(G')$ such that,
for all $i,k\in[a]$ and $j,l\in[n]$,  an edge $\{u^{(i)}_j,u^{(k)}_l\}$ is in $E(G')$ if either (1) $i=j$ and $j=l$, or (2) $i\neq j$ and $\{v_j,v_l\}\in E(G)$ holds.
The set $E(G')$ does not contain any other edges.
This graph can be constructed in time $O(an^2)$.

Now we check the correctness.
Suppose that a vertex set $S=\{v_{i_1},\ldots,v_{i_a}\}$ forms an $a$-clique in $G$.
Then, the vertex set $\{u^{(1)}_{i_1},\ldots,u^{(a)}_{i_a},w^{(1)}_{i_1},\ldots,w^{(a)}_{i_a}\}$ forms a colorful $K_{a,a}$-subgraph in $G'$.
Conversely, if the set $\{u^{(1)}_{i_1},\ldots,u^{(a)}_{i_a},w^{(1)}_{j_1},\ldots,w^{(a)}_{j_a}\}$ forms a colorful $K_{a,a}$-subgraph in $G'$, then it holds that $i_1=j_1,\ldots,i_a=j_a$ and the set $\{v_{i_1},\ldots,v_{i_a}\}$ forms an $a$-clique in $G$.
\end{proof}

\subsection{An \texorpdfstring{$n^{a+o(1)}$}{}-Time Algorithm for \texorpdfstring{$\#\mathsf{EMB}^{(K_{a,b})}$}{Counting Kab}} \label{sec:fastalgo}

We now present an algorithm that matches the lower bounds presented so far.
Specifically, we design an algorithm that solves $\numberembcol{K_{a, b}}$ in time $O(b n^{a + o(1)})$,
thereby proving \cref{prop:fastalgo}.
The algorithm of \cref{prop:fastalgo} is similar to the $O(n^{k+o(1)})$-time algorithm for \textsc{$k$-Dominating Set} of $k\geq 7$ proposed by P\v{a}tra\c{s}cu and Williams~\cite{PW10}.
The algorithm of \cite{PW10} adopts the fast rectangular matrix multiplication~\cite{LU18,LeGall12,Cop97}.
Recently, Le Gall and Urrutia~\cite{LU18} proved that we can compute the multiplication of an $n\times n^\gamma$ matrix and an $n^\gamma\times n$ matrix in $n^{2+o(1)}$ arithmetic operations if $\gamma\leq 0.31389$.

Let $G=(V,E)$ be a given instance of $\KabCount$.
We first consider the case when $a$ is even.
We construct an $\binom{n}{a/2}\times n$ matrix $B$ as follows:
For each $S\in\binom{V}{a/2}$ and $v\in V$,
\begin{align*}
    B[S][v]=
    \begin{cases}
    1 & \text{if $S \subseteq N(v)$},\\
    0 & \text{otherwise}.
    \end{cases}
\end{align*}
Then, compute the product $BB^{\top}$ by the fast rectangular matrix multiplication~\cite{LU18}.
The running time is $O(n^{a+o(1)})$ if $a\geq 8$.
Notice that $BB^{\top}[S_1][S_2]$ is equal to the size of the vertex subset $W(S_1,S_2)$, where $W(S_1,S_2)\defeq \{v\in V\setminus (S_1\cup S_2): S_1\cup S_2\subseteq N(v)\}$.
In other words, the set $W(S_1,S_2)$ contains vertices that is adjacent to all vertices in $S_1\cup S_2$.
For any $S_1,S_2\in \binom{V}{a/2}$ with $S_1\cap S_2=\emptyset$ and $T\in \binom{W(S_1,S_2)}{b}$, the vertex set $S_1\cup S_2 \cup T$ forms a $K_{a,b}$ subgraph.
On the other hand, for a $K_{a,b}$ subgraph, there are $c\binom{a}{a/2}$ ways to take $S_1,S_2,T$, where $c=2$ if $a<b$ and $c=4$ if $a=b$.
If $a<b$, the factor $c$ reflects the symmetry of $S_1$ and $S_2$; thus $c=2$.
If $a=b$, we further take the symmetry of $S_1\cup S_2$ and $T$ into account; thus $c=4$.
Then, the number of $K_{a,b}$ subgraphs contained in $G$ is given by
\begin{align*}
    c^{-1} \binom{a}{a/2}^{-1}\cdot \sum_{S_1,S_2\in\binom{V}{a/2}:S_1\cap S_2=\emptyset}\binom{BB^{\top}[S_1][S_2]}{b}.
\end{align*}

Now consider the case when $a$ is odd.
Fix a vertex $u\in V$.
Again, we construct an $\binom{n}{(a-1)/2}\times n$ matrix $B^{(u)}$ as follows:
For each $S\in\binom{V}{(a-1)/2}$ and $v\in V$,
\begin{align*}
    B^{(u)}[S][v] = \begin{cases}
1 & \text{if $\{u,v\}\in E,v\not\in S$ and $S\subseteq N(v)$},\\
0 & \text{otherwise}.
    \end{cases}
\end{align*}
Then compute $B^{(u)}(B^{(u)})^{\top}$ for all $u\in V$.
Note that the multiplication can be computed in
time $n^{a-1+o(1)}$ for each $u\in V$.
Observe that $B^{(u)}(B^{(u)})^{\top}[S_1][S_2]$ is the number of vertices that is adjacent to all vertices in $S_1\cup S_2\cup \{u\}$.
Thus, the number of $K_{a,b}$ contained in $G$ is given by
\begin{align*}
    c^{-1} \left(a\binom{a-1}{(a-1)/2}\right)^{-1}\cdot \sum_{u\in V}\sum_{S_1,S_2\in \binom{V}{(a-1)/2}:S_1\cap S_2=\emptyset} \binom{B^{(u)}(B^{(u)})^{\top}[S_1][S_2]}{b},
\end{align*}
where $c=2$ if $a<b$ and $c=4$ if $a=b$.

This yields an $O(bn^{a+o(1)})$ time algorithm (note that $\binom{n}{k}$ can be computed in $O(k\log n)$ time).

\subsection{Reduce \texorpdfstring{$(\#\mathsf{EMB}_{\mathrm{col}}^{(K_{a,b})}, \distribution{K_{a,b}})$}{} to \texorpdfstring{$(\#\mathsf{EMB}^{(K_{a,b})},\DistKab{a}{b}{n})$}{Kabn}}
\label{subsec:step3}

In this section, we present a proof of 
\cref{prop:ToDistKab}, i.e.,
an average-case-to-average-case reduction from $(\numberembcol{K_{a,b}},\distribution{K_{a,b}})$
to $(\numberemb{K_{a,b}},\DistKab{a}{b}{n})$.
This will complete a proof of \Cref{thm:Kabworsttoaverage}:
Recall that \cref{thm:worst-to-average_precise}
reduces $\numberembcol{K_{a,b}}$ to $(\numberembcol{K_{a,b}},\distribution{K_{a,b}})$.
Combined with \cref{folk:uncolortocolor},
one can reduce $\numberemb{K_{a,b}}$ to $(\numberembcol{K_{a,b}},\distribution{K_{a,b}})$.
Overall, we obtain a reduction from 
$\numberemb{K_{a,b}}$ to $(\numberemb{K_{a,b}},\DistKab{a}{b}{n})$ as stated in \cref{thm:Kabworsttoaverage}.

%


\begin{proof}[Proof of \cref{prop:ToDistKab}]
Let $\randombigraphdist{n}{m}$ be the distribution of a random bipartite graph with left and right vertex sets of size $n$ and $m$, respectively.
Let $G$ be an input of $(\numberembcol{K_{a,b}},\distribution{K_{a,b}})$.
Observe that the distribution $\distribution{K_{a,b}}$ is identical to $\randombigraphdist{an}{bn}$.
We say that a subgraph $F\subseteq G$ \emph{contains color $i$} if $F$ contains a vertex of color $i$.
Let $S$ be the set of subgraphs $F\subseteq G$ isomorphic to $K_{a,b}$.
Let $S_i\subseteq S$ be the set of subgraphs $F \in S$ that contain color $i$.
Observe that $\embcol{K_{a,b}}{G}=\left|\bigcap_{i\in V(K_{a,b})} S_i\right|$.
By the inclusion-exclusion principle, we have
\begin{align*}
    \left|\bigcap_{u\in V(K_{a,b})} S_u\right| &= |S| - \left|\bigcup_{i\in V(K_{a,b})} \overline{S_i}\right| \\
    &= |S| - \sum_{\emptyset\neq J\subseteq V(K_{a,b})} (-1)^{|J|-1} \left| \bigcap_{j\in J} \overline{S_j}\right|.
\end{align*}
In light of this equality, it suffices to compute $|S|$ and $\left|\bigcap_{j\in J} \overline{S_j}\right|$ for all nonempty $J\subseteq V(K_{a,b})$.
Note that the set $\cap_{j\in J}\overline{S_j}$ is equal to the set of $K_{a,b}$ subgraphs in $G$ that does not contain any colors from $J$.
To state it more formally, for a nonempty set $J\subseteq V(K_{a,b})$, let $V_J=\{x\in V(G):c(x)\in J\}$ and $G_J=G[V_J]$ be the induced subgraph of $G$ by $V_J$.
Then, $\bigcap_{j\in J}\overline{S_j}$ is equal to the set of $K_{a,b}$ subgraphs contained in $G_{\overline{J}}$.
Suppose that we have a $T(n)$-time randomized algorithm $A$ that solves $(\numberemb{K_{a,b}},\DistKab{a}{b}{n})$ with failure probability $\delta$.
Note that, for each $J\subseteq V(K_{a,b})$, 
the distribution of  $G_{\overline{J}}$ for $G \sim \randombigraphdist{an}{bn}$
is identical to
$\randombigraphdist{cn}{dn}$ for some $c\leq a$ and $d\leq b$; thus, we can obtain $\left|\bigcap_{j\in J}\overline{S_j}\right|$ with probability at least $1-ab\delta$ since $A(G_{\overline{J}})=c!d!\left|\bigcap_{j\in J}\overline{S_j}\right|$
(here, $c!d!$ is the number of automorphisms of $K_{c,d}$).
Therefore, from the union bound,
we can obtain $\left|\bigcap_{j\in J}\overline{S_j}\right|$ for all $\emptyset\neq J\subseteq V(K_{a,b})$ with probability at least $1-ab2^{a+b}\delta$.
Moreover, $A(G)=a!b!|S|$ holds with probability $1-ab\delta$.
Hence, we can solve $(\numberemb{K_{a,b}},\distribution{K_{a,b}})$ in time $O(ab2^{a+b}\cdot T(n))$ with probability $1-O(ab2^{a+b} \delta)$.
\end{proof}

\subsection{Proofs of \texorpdfstring{\cref{thm:KabCountSETH,prop:fastalgo,thm:Kabworsttoaverage}}{Kab Results}}
\label{sec:proof_of_Kab_results}

\Cref{thm:Kabworsttoaverage} follows from \cref{folk:uncolortocolor,thm:worst-to-average_precise,prop:ToDistKab}.
We can show \cref{thm:KabCountSETH} by combining \cref{thm:mainthm1,lem:coluncolequivalent} (Note that we can solve $\ColKabDetect$ using a solver for $\KabCount$).
Similarly, \cref{thm:KaaCountETH} follows from \cref{lem:Ka_to_Kaa_reduction,lem:coluncolequivalent}
and the well-known fact that ETH rules out an $n^{o(a)}$-time algorithm for $\KaDetect$~\cite{CHKX06}.
\section{Doubly-Efficient Interactive Proof System} \label{sec:doubly_efficient_IP}
This section is devoted to proving \cref{thm:embcolIP}.
Recall that, in $\numberembcol{H}$, we are given $n$ and $G\subseteq K_n\times H$ and then are asked the number of colorful subgraphs $F\subseteq G$ that is isomorphic to a fixed graph $H$.
Fix a prime $n^{|V(H)|}<q<2n^{|V(H)|}$ and consider the polynomial $\embcolpoly{n}{H}{q}:\Fq^{E(K_n\times H)} \to \Fq$ defined in \cref{eq:embcolpolydef}.
In our interactive proof system $\IP$,
the verifier checks the statement that
$\embcolpoly{n}{H}{q}(x)=C$ for given $C\in\Fq$ and $x\in\Fq^{E(K_n\times H)}$.
Recall that, if $x$ is an edge indicator vector of a graph $G$,
then $\embcolpoly{n}{H}{q}(x)=\emb{H}{G}$ holds.

\subsection{Downward Reducibility}
We may assume without loss of generality that 
$n=2^t$ for some $t\in\Nat$.
(If not, we add isolated vertices to $G$.)
For each $i\in V(H)$, let $V_i\defeq \{v\in V(K_n\times H):c(v)=i\}$.

For each $i\in V(H)$, let $(V_{i,0},V_{i,1})$ be a partition of $V_i$ such that $|V_{i,0}|=|V_{i,1}|=|V_i|/2$. 
For $\eta\in\binset^{V(H)}$, let $E_{\eta} = \cup_{ij\in E(H)} E(V_{i,\eta(i)},V_{j,\eta(j)})$, where $E(S,T)=\{e\in E(K_n\times H): e\cap S\neq \emptyset \text{ and }e\cap T\neq \emptyset\}$ for $S,T\subseteq V(K_n\times H)$ (see \cref{fig:example}).
Since $|V_{i,\eta(i)}|=|V_i|/2$, 
we can identify $E_\eta$ with $E(K_{n/2}\times H)$.
From the definition~\cref{eq:embcolpolydef}, we have
\begin{align}
\embcolpoly n H q (x) = \sum_{\eta\in\binset^{V(H)}} \embcolpoly{n/2}{H}{q} (x[E_\eta]), \label{eq:embcolpoly_recursion1}
\end{align}
where $x[E_\eta]\in\Fq^{E_\eta}$ is the restriction of $x$ on $E_\eta$.

\begin{figure}[htbp]
\centering
\includegraphics[width=10cm]{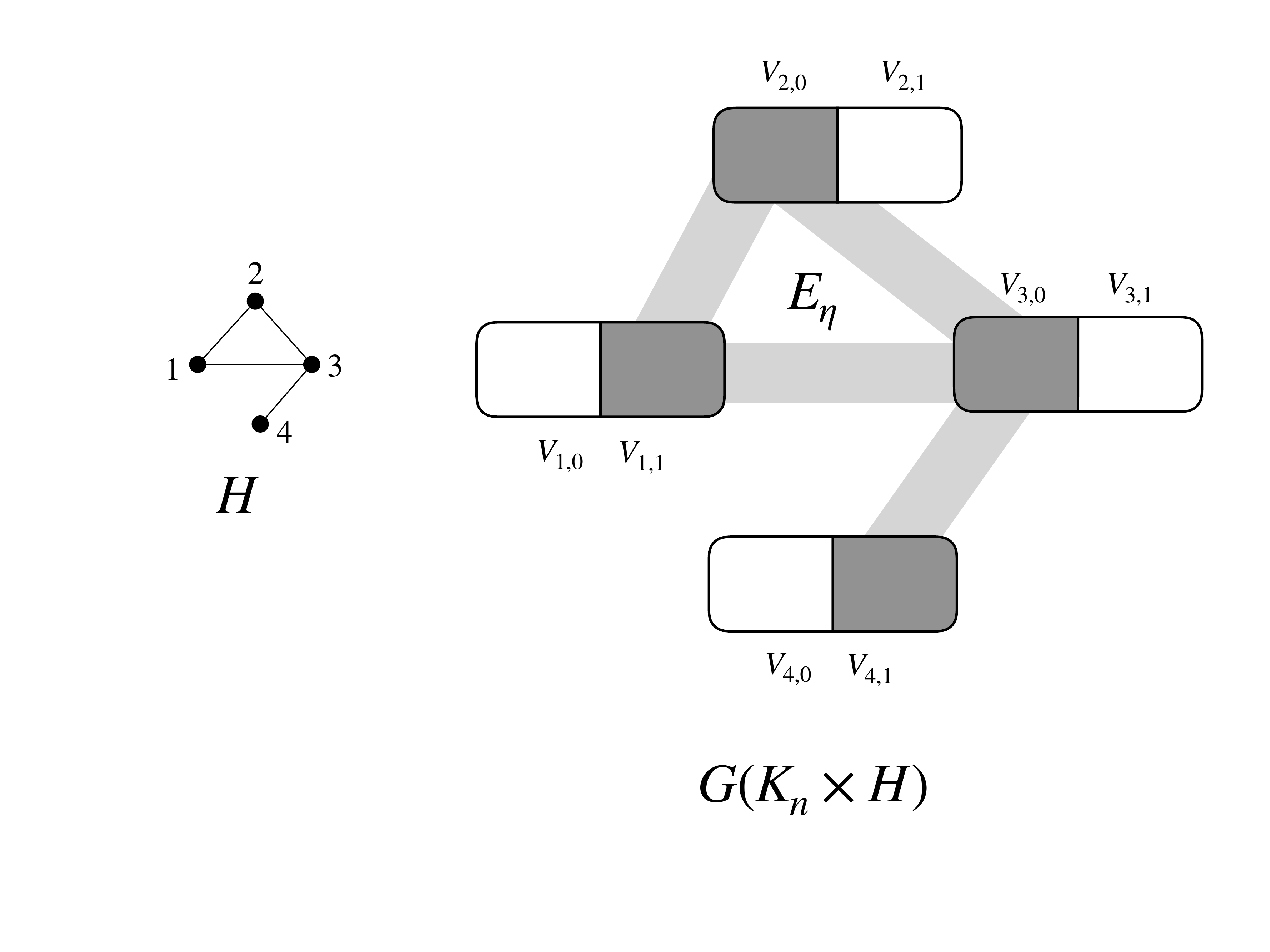}
\caption{An example of $E_\eta$.
In this example, $\eta=(1,0,0,1)\in \binset^{V(H)}$.
Four grey areas represent $V_{i,\eta(i)}$ for $i=1,2,3,4$.
\label{fig:example}}
\end{figure}

We identify $\{0,1\}^{V(H)}$ with $\{0,1,\ldots,2^{|V(H)|}-1\} \subseteq \Fq$
in the following way.
Regard $V(H)=\{0,\ldots,k-1\}\subseteq \Fq$ for $k = |V(H)|$ and
consider the mapping $\{0,1\}^{V(H)}\ni \eta \mapsto \sum_{i\in V(H)} 2^i \eta(i)\in \{0,\ldots,2^{|V(H)|}-1\}\subseteq \Fq$.
This mapping is bijection and thus
we can regard $\eta$ as an element of $\Fq$.
For $\eta\in \{0,\ldots,2^{|V(H)|}-1\}$, let $\delta_\eta:\Fq\to\Fq$ be a degree-$(2^{|V(H)|}-1)$ polynomial
such that
\begin{align*}
\delta_\eta(z)=\begin{cases}
1 & \text{if $z=\eta$},\\
0 & \text{if $z\neq \eta$ and $0\leq z<2^{|V(H)|}$},\\
\text{arbitrary} & \text{otherwise}.
\end{cases}
\end{align*}
Then, define a function 
$\tilde{x}:\Fq\to\Fq^{E(K_{n/2}\times H)}$ by
$\tilde{x}(\cdot)\defeq \sum_{\eta\in\binset^{V(H)}} \delta_\eta(\cdot) x[E_\eta]$.
Note that $\tilde{x}$ satisfies
(1) $\tilde{x}(\eta)=x[E_\eta]$ holds for all $\eta\in\binset^{V(H)}$,
and
(2) for each $e\in E(K_{n/2}\times H)$,
the function
$\Fq \ni z \mapsto \tilde{x}(z)[e]\in \Fq$ is
a polynomial of degree $2^{|V(H)|}-1$.
In condition (1),
we identified $E_\eta$ with
$E(K_{n/2}\times H)$.
For each $\eta$, the function $\delta_\eta$ can be constructed
by $O(2^{|V(H)|}|V(H)|)$ field operations
using the fast univariate polynomial interpolation~\cite{Hor72}.
Since our computational model is $O(\log n)$-Word RAM,
we can perform any field operation on $\Fq$ in constant time.
Thus, the construction of $\tilde{x}$ can be done in
time $O(2^{|V(H)|}|V(H)||E(H)|n^2)$.
Using $\tilde{x}(\cdot)$, we can
rewrite the recursion formula~\cref{eq:embcolpoly_recursion1}
as
\begin{align}
\embcolpoly{n}{H}{q}(x)=\sum_{\eta\in\{0,\ldots,2^{|V(H)|}-1\}} \embcolpoly{n/2}{H}{q}(\tilde{x}(\eta)). \label{eq:embcolpolydownwardreducibility}
\end{align}
Note that $\embcolpoly{n}{H}{q}(\tilde{x}(\cdot))$ is a univariate polynomial of degree $|E(H)|(2^{|V(H)|}-1)$.

\subsection{Description and Analysis of \texorpdfstring{$\IP$}{IP}} \label{sec:ip}
Now we present $\IP$ that verifies the statement ``$\embcolpoly{n}{H}{q}(x)=C$".
Suppose that $n=2^t$.
$\IP$ consists of $t+1$ rounds.
The verifier is given a vector $x\in\Fq^{E(K_n\times H)}$ (in our case, $x$ is the edge indicator of the input graph $G$).
At each round,
the verifier updates the vector $x$ and the constant $C$.
In $r$-th round, the protocol proceeds as follows.

\begin{description}
\item[Verifier.] When $r=t+1$, check $\embcolpoly{n}{H}{q}(x)=C$ and halt.
\item[Prover.] Send a polynomial $G(\cdot)$ of degree at most $|E(H)|(2^{|V(H)|}-1)$ over $\Fq$ to the verifier (here, $G(\cdot)$ is expected to be the univariate polynomial $\embcolpoly{n/2}{H}{q}(\tilde{x}(\cdot))$, where $\tilde{x}$ is the polynomial vector constructed from $x$).
\item[Verifier.] Check $C=\sum_{z\in\{0,\ldots,2^{|V(H)|}-1\}}G(z)$.
If not, reject.
Otherwise, 
construct the polynomial vector $\tilde{x}(\cdot)$ using $x$.
Sample $i\sim\Unif{\Fq}$
and update
$x\leftarrow \tilde{x}(i)$,
$C\leftarrow G(i)$ and $n\leftarrow n/2$.
Then, proceed to the next round
recursively.
\end{description}

The task of the prover is to construct the function $\embcolpoly{n/2}{H}{q}(\tilde{x}(\cdot))$.
Note that we can construct
$\embcolpoly{n/2}{H}{q}(\tilde{x}(\cdot))$ by evaluating $\embcolpoly{n/2}{H}{q}(\tilde{x}(\cdot))$ at $|E(H)|(2^{|V(H)|}-1)+1$ points.
The evaluation is reducible to $\numberembcol{H}$ via \cref{eq:EMBCOLpolybinaryexpand}.
Thus, one can modify $\IP$ above
such that the verifier asks the prover
to solve $\numberembcol{H}$
for $(\log n)^{O(|E(H)|)}$ instances.
In what follows,
we analyze this modified protocol.

\subsubsection{Running Time}
Let $n$ be the size of the original input.
In the beginning of $r$-th round,
the size of $x$ is $|E(K_n\times H)|/4^{r-1}=|E(H)|\cdot 4^{t-r+1}$.
Thus, in the $(t+1)$-th round,
the verifier runs in constant time.
Any other task of the verifier can be
done in time $n^2(\log n)^{O(|E(H)|)}$
(the bottleneck is the simulation of
the reduction of constructing $\embcolpoly{n/2}{H}{q}(\tilde{x}(\cdot))$
to $\numberembcol{H}$).

\subsubsection{Completeness and Soundness}
The perfect completeness of $\IP$ is easy:
If the statement is true, an honest prover convinces the verifier with probability $1$ by sending the polynomial $\embcolpoly{n/2}{H}{q}(\tilde{x}(\cdot))$ at each round (recall that the polynomial $\embcolpoly{n}{H}{q}(\tilde{x}(\cdot))$ satisfies the recurence formula ~\cref{eq:embcolpolydownwardreducibility}).

Now we show the soundness.
\begin{proposition}\label{prop:soundness}
Let $n=2^t$.
If the statement ``$\embcolpoly{n}{H}{q}(x)=C$" is false, then for any provers, the verifier rejects with probability $\left(1-\frac{D}{q}\right)^{t-1}$.
\end{proposition}
\begin{proof}
We show \cref{prop:soundness} by induction on the number of rounds.
Suppose that the statement ``$\embcolpoly{n}{H}{q}(x)=C$" is false.
In the last $t$-th round, the verifier immediately reject (with probability $1$).

Let $D\defeq |E(H)|(2^{|V(H)|}-1)$ be the degree of $\embcolpoly{n/2}{H}{q}(\tilde{x}(\cdot))$.
Fix $1\leq r<t$ and suppose that the verifier rejects with probability $(1-D/q)^{t-r}$ during $(r+1)$-th to $t$-th rounds.
Consider $r$-th round with the assumption that the the statement ``$\embcolpoly{n}{H}{q}(x)=C$" is false.
Then, any provers cheat with sending a polynomial $G(\cdot)$ that is not equal to $\embcolpoly{n/2}{H}{q}(\tilde{x}(\cdot))$.
It holds that
\begin{align*}
     \Pr_{i\sim\Unif{\Fq}}\left[G(i)\neq \embcolpoly{n/2}{H}{q}(\tilde{x}(i))\right] \geq 1-\frac{D}{q}
\end{align*}
since both $G$ and $\embcolpoly{n/2}{H}{q}(\tilde{x}(\cdot))$ are degree at most $D$.
Therefore, with probability $1-D/q$, the rounds proceeds to the next $(t+1)$-th round with a false statement.
By the induction assumption, the verifier rejects with probability at least $(1-D/q)\cdot (1-D/q)^{t-r} = (1-D/q)^{1+t-r}$, which completes the proof of \cref{prop:soundness}.
\end{proof}

\subsubsection{Ability of an Honest Prover}
In $\IP$, the prover is required to send a polynomial $G$ that is expected to be $\embcolpoly{n/2}{H}{q}(\tilde{x}(\cdot))$.
As mentioned above, this can be reduced to
solving 
$m=\polylog(n)$ instances of $\numberembcol{H}$.
Moreover, by \cref{thm:worst-to-average_precise},
each of the $m$ instances 
can be reduced to $\polylog(n)$ instances
of $(\numberembcol{H},\distribution{H})$.
In other words, an honest prover
can construct the polynomial $\embcolpoly{n/2}{H}{q}(\cdot)$ by
solving $m\polylog(n) = \polylog(n)$ instances
of the distributional problem $(\numberembcol{H},\distribution{H})$.

Suppose that the prover has oracle access to a randomized heuristic
algorithm that solves $(\numberembcol{H},\distribution{H})$
with success probability $1-(\log n)^{-L_1}$.
Then, by the union bound,
the probability that the oracle outputs
at last one wrong answer
is at most
$\polylog(n)\cdot (\log n)^{-L_1} \leq (\log n)^{-L_0}$
if $L_1=L_1(H,L_0)$ is
sufficiently large.

This completes the proof of \cref{thm:embcolIP}.

\subsection{Interactive Proof System for \texorpdfstring{$\#\mathsf{EMB}^{(K_{a,b})}$}{Kab Counting}}
By combining \cref{thm:Kabworsttoaverage,lem:coluncolequivalent,thm:embcolIP},
we obtain an interactive proof system for $\KabCount$ as follows.

\begin{corollary}[IP for $\KabCount$]
\label{cor:ip_for_kabcounting}
Let $H$ be a fixed graph.
There is an $O(\log n)$-round interactive proof system $\IP$ 
for the statement ``$\emb{K_{a,b}}{G}=C$"
such that,
given an input $(G,n,C)$,
\begin{itemize}
    \item The verifier accepts with probability $1$ if the statement is true (perfect completeness), while it rejects for any prover
    with probability at least $2/3$ otherwise (soundness).
    \item In each round, the verifier runs in time $n^2(\log n)^{O(ab)}$ and sends $(\log n)^{O(ab)}$ instances of $\numberemb{K_{a,b}}$.
\end{itemize}
Furthermore, for any contant $L_0$,
there exists a constant $L_1=L_1(a,b,L_0)$
such that,
if the prover solves
$(\KabCount,\DistKab a b n)$
with success probbalitiy
$1-(\log n)^{-L_1}$,
then the verifier accepts
with probability $1-(\log n)^{-L_0}$.
\end{corollary}
\begin{proof}
For a given graph $G$, the verifier applies \cref{lem:coluncolequivalent}
and reduces $\KabCount$
to solving $m=O(1)$ instances
of $\numberembcol{K_{a,b}}$.

Then, the verifier solves
each of the $m$ instances 
$G'_1,\ldots,G'_{m}$ of $\numberembcol{K_{a,b}}$
using the reduction of \cref{thm:worst-to-average_precise}
with the help of the prover.
Let $C'_1,\ldots,C'_m$ be the values
obtained by the reduction.
Now, the verifier
suffices to check that,
for each $i\in[m]$,
the answer of the $i$-th instance $G'_i$ of
$\numberembcol{K_{a,b}}$ is
$C'_i$
(if all of these $m$ values are
the correct one,
then the verifier could
solve the original
instance $G$ of $\KabCount$).

To this end, run $\IP$
of \cref{thm:embcolIP} with
letting $H=K_{a,b}$.
Here, an honest prover
suffices to solve $(\KabCount,\DistKab a b n)$
since
$\numberembcol{K_{a,b}}$
reduces to
solving $\polylog(n)$
instances of
$(\KabCount,\DistKab a b n)$
by combining the reductions of
\cref{lem:coluncolequivalent,thm:Kabworsttoaverage}.
\end{proof}

\section{Fine-Grained Direct Product Theorem} \label{sec:fine-grained_direct_product_theorem}

\newcommand{\Preprocessor}{\mathsf{Preprocessor}}
\newcommand{\Checker}{\mathsf{Checker}}

In this section, we provide a sufficient condition 
for a direct product theorem to hold.
We will present a direct product theorem
for any distributional problem that admits a selector.
The notion of selector that we use in this paper is defined below.

\begin{definition}
[(Oracle) Selector; \cite{Hirahara15_coco_conf}]
A randomized oracle algorithm $S$ is said to be a \emph{selector} from $\Pi$ to a distributional problem $(\Pi, \mathcal{D})$ with success probability $1 - \delta$
if
\begin{enumerate}
    \item 
    given access to two oracles $A_0, A_1$ one of which
    solves $(\Pi, \mathcal{D})$ with success probability $1 - \delta$,
    on input $x$,
    the algorithm 
    $S^{A_0, A_1}$
    computes $\Pi(x)$ with high probability (say, probability $\ge \frac{3}{4}$), and
    \item
    for any $n\in\Nat$ and any input $x \in \supp(\mathcal{D}_n)$,
    each query $q$ of $S$ to the oracles $A_0$ and $A_1$
    satisfies that $q \in \supp(\mathcal{D}_n)$.
\end{enumerate}

\end{definition}

In order to obtain
a direct product theorem 
in the settings of fine-grained complexity,
it will be crucial to consider a selector with
$\polylog(n)$ queries.

\subsection{Selector for Subgraph Counting Problems}
\label{subsection:SelectorForCounting}
In this subsection, 
we show the existence of a selector with $\polylog(n)$ queries
for $\numberembcol{H}$.
We first recall the notion of instance checker, which is known to imply the existence of selector (\cite{Hirahara15_coco_conf}).

\begin{definition}
[Instance Checker; Blum and Kannan \cite{BlumK95_jacm_journals}] \label{def:InstanceChecker}
For a problem $\Pi$,
a randomized oracle algorithm $M$ is said to be an \emph{instance checker} for $\Pi$
if, for every instance $x$ of $\Pi$ and any oracle $A$,
\begin{enumerate}
\item
$\Pr_M[ M^A(x) = \Pi(x) ] = 1$ if $A$ solves $\Pi$ correctly on every input, and
\item
$\Pr_M[ M^A(x) \not\in \{\Pi(x), \mathsf{fail}\} ] = o(1)$.
\end{enumerate}
\end{definition}
\noindent

The existence of an instance checker for a problem $\Pi$ is implied by an efficient interactive proof system
for $\Pi$
where the computation of an honest prover is efficiently reducible to $\Pi$.
By using the interactive proof system of \cref{thm:embcolIP}, we obtain
the following
instance checker with $\polylog(n)$ queries for $\numberembcol{H}$.

\begin{theorem}\label{thm:InstanceCheckerForEmbcol}
There exists an instance checker $\Checker$ for $\numberembcol{H}$
such that, given a graph $G \subseteq K_n \times H$,
\begin{enumerate}
	\item $\Checker$ runs in time $\widetilde{O}(n^2)$,
	\item for any oracle $A$, $\Checker^A$ calls the oracle $A$ at most $(\log n)^{C_H}$ times, where $C_H$ is a constant that depends only on $H$, and
	\item each query $G'$ of $\Checker$ satisfies $G'\subseteq K_n\times H$.
\end{enumerate}
\end{theorem}
\begin{proof}
Recall $\IP$ of \cref{thm:embcolIP}.
For a given oracle $A$, $\Checker$ obtains $C\defeq A(G)$ and then runs $\IP$ using $A$ as a prover to verify $\embcol{H}{G}=C$.
If the verifier accepts, then $\Checker$ outputs $C$ and otherwise it outputs $\mathsf{fail}$.

Suppose the oracle $A$ solves $\numberembcol{H}$ correctly.
Then, $\Checker$ output the correct answer with probability $1$ by the perfect completeness of $\IP$.

Now we check the second condition of \cref{def:InstanceChecker}.
If $A(G)$ is the correct answer, then the output of $\Checker$ is either $A(G)$ or $\mathsf{fail}$.
Otherwise, $\IP$ proceeds with the false statement that $\numberembcol{H}(G)=A(G)$.
It follows from the soundness of $\IP$ (c.f.~\cref{prop:soundness}) that Verifier rejects with probability $(1-O(q^{-1}))^t$.
Hence, $\Pr[\Checker(G)=\mathsf{fail}]\geq 1-O(t/q)=1-o(1)$.
\end{proof}

\begin{theorem}[Restatement of \cref{thm:selector}]
\label{thm:SelectorForEmbcolH}
Let $H$ be a fixed graph.
There exists a selector $S$
from $\numberembcol{H}$ to $(\numberembcol{H}, \distribution{H})$ with success probability $1 - 1 / \polylog(n)$ such that
\begin{enumerate}
\item $S$ runs in time $\widetilde{O}(n^2)$, and
\item $S$ makes at most $\polylog(n)$ queries.
\end{enumerate}
\end{theorem}
\begin{proof}
We combine the instance checker $C$ 
of
\cref{thm:InstanceCheckerForEmbcol}
with a worst-case to average-case reduction  $R$ (\cref{thm:worst-to-average_precise}).

Here is the algorithm of a selector $S$.
Given a graph $G$ and oracle access to $A_0, A_1$,
for each $b \in \{0, 1\}$,
the selector $S$ simulates the instance checker $C(G)$,
and answer any query $q$ of the instance checker by running the reduction $R^{A_b}(q)$.
If the checker outputs some answer other than $\mathsf{fail}$,
the selector $S$ outputs the answer and halts.

The correctness of $S$ can be shown as follows.
Let $A_b$ be an oracle that solves 
$(\numberembcol{H}, \distribution{H})$
with success probability $1 - 1 / \polylog(n)$,
where $b \in \{0, 1\}$.
By the correctness of the reduction $R$,
the algorithm $R^{A_b}$ solves $\numberembcol{H}$
with high probability.
Therefore,
if 
the instance checker $C$ is simulated with oracle access to $R^{A_b}$,
by the property of an instance checker,
$C$ outputs the correct answer with high probability.
Moreover, $C$ outputs a wrong answer with probability at most $o(1)$;
thus, the selector outputs the correct answer with high probability.
\end{proof}

\begin{corollary}
\label{cor:SelectorForKabCount}
There is an $\widetilde{O}(n^2)$-time selector $S$
from $\KabCount$ to $(\KabCount, \DistKab{a}{b}{n})$ with success probability $1 - 1 / \polylog(n)$.
Moreover, 
$S$ makes at most $\polylog(n)$ queries.
\end{corollary}
\begin{proof}
From \cref{thm:SelectorForEmbcolH}, we obtain a selector $S$
from $\numberembcol{K_{a, b}}$ to $(\numberembcol{K_{a, b}}, \allowbreak \distribution{K_{a,b}})$.
Here, we let $H=K_{a,b}$.
Invoke \cref{prop:ToDistKab} that reduces 
$(\numberembcol{K_{a, b}}, \distribution{K_{a,b}})$
to
$(\KabCount,\DistKab{a}{b}{n})$ in $2^{O(a+b)}=O(1)$ time.
Note that the reduction of \cref{prop:ToDistKab} preserves the success probability within a constant factor.
Thus, each oracle query of $S$ can be replaced by the reduction and we obtain a selector
from
$\numberembcol{K_{a, b}}$ to $(\KabCount, \DistKab{a}{b}{n})$.
By \cref{lem:coluncolequivalent},
$\KabCount$ is efficiently reducible to $\numberembcol{K_{a, b}}$, from
which the existence of a selector
from
$\KabCount$  to $(\KabCount, \DistKab{a}{b}{n})$ follows.
%
\end{proof}

\subsection{Direct Product Theorem for Any Problem with Selector}
\label{subsection:DirectProductAndSelector}

Using the notion of selector, we provide 
a direct product theorem in the context of fine-grained complexity.
A direct product of a distributional problem is formally defined as follows.

\begin{definition}
[Direct Product] \label{def:direct_product}
Let $k \colon \Nat \to \Nat$ be any function, and $(\Pi, \mathcal{D})$ be any distributional problem.
The \emph{$k$-wise direct product of $(\Pi, \mathcal{D})$},
denoted by $(\Pi, \mathcal{D})^k$,
is defined as the distributional problem $(\Pi^k, \mathcal{D}^k)$
such that
\begin{enumerate}
    \item
      $(\mathcal{D}^k)_n \defeq \mathcal{D}_n^{k(n)}$ for each $n \in \Nat$, and
    \item 
        $\Pi^k(x_1, \cdots, x_{k(n)}) \defeq (\Pi(x_1), \cdots, \Pi(x_{k(n)}))$
        for any $(x_1, \cdots, x_{k(n)}) \in \supp ( \mathcal{D}^{k(n)}_n )$.
\end{enumerate}
\end{definition}


The following direct product theorem gives an \emph{almost} uniform direct product,
in the sense that it requires $O(\log 1 / \epsilon)$ bits of non-uniform advice
in order to identify which is a correct algorithm.
We observe that the direct product theorem is quite efficient
and useful even in the setting of fine-grained complexity.
\begin{theorem}
	[Almost Uniform Direct Product; Impagliazzo, Jaiswal, Kabanets, and Wigderson~\cite{ImpagliazzoJKW10_siamcomp_journals}]
	\label{thm:AlmostUniformDirectProduct}
	Let $k \in \Nat$, $\epsilon, \delta > 0$ be parameters that satisfy $\epsilon > \exp{(-\Omega(\delta k))}$.
	There exists a randomized oracle algorithm $M$ 
	that, given access to an oracle $C$
	that solves $(\Pi, \mathcal{D})^k$ with success probability $\epsilon$,
	with high probability,
	produces a list of deterministic oracle algorithms $M_1, \cdots, M_m$
	such that 
	$M_i^C$ computes $(\Pi, \mathcal{D})$ with success probability $\delta$ for some $i \in \numset{m}$,
	where $m = O(1 / \epsilon)$.
	
	If an oracle $C$ can be computed in $T_C(n)$ time,
	then
	the running time of $M_i^C$ is at most $T_C(n) \cdot O( (\log 1 / \delta) / \epsilon)$
	for any $i$;
	the running time of $M$ is at most $O(T_C(n) / \epsilon)$.
\end{theorem}

\begin{proof}
[Proof Sketch]
The algorithm $M^C$ operates as follows.
Fix an instance size $n \in \Nat$.
Let $U$ denote $\supp(\mathcal{D}_n)$.
Repeat the following $m = O(1 / \epsilon)$ times.
Pick a random (ordered) subset $B_0 \subset U$ of size $k$,
and  pick a random ordered subset $A \subset B_0$ of size $k / 2$.
Evaluate $C(B_0)$ and let $v$ be the answers given by $C(B_0)$ for the instances in $A$.
Output an oracle algorithm $M_{A, v}$ defined below.

The algorithm $M_{A, v}^C$ is defined as follows.
On input $x \in U$, check whether $x = a_i$ for some $a_i \in A$; if so, output $v_i$.
Otherwise, repeat the following $O((\log 1 / \delta) / \epsilon)$ times.
Sample a random set $B \supset A \cup \{x\}$ of size $k$.
(The randomness used here can be supplied by $M^C$ and thus $M_{A, v}^C$ can be made deterministic.)
If the answers given by $C(B)$ for the instances in $A$ coincide with $v$,
then output the answer given by $C(B)$ for the instance $x$ and halt;
otherwise, go to the next loop.
\end{proof}

\begin{lemma} \label{lem:solve_by_selector}
Let $(\Pi,\mathcal{D})$ be a distributional problem.
Suppose there exists a selector $S$ from $\Pi$ to $(\Pi,\mathcal{D})$ with success probability $\delta$ that calls an oracle at most $\polylog(n)$ times.
Let $M_1,\ldots,M_m$ be a list of deterministic algorithms such that,
(1) for some $i\in\{1,\ldots,m\}$, $M_i$ solves $(\Pi,\mathcal{D})$ with success probability $\delta$,
(2) for all $i\in\{1,\ldots,m\}$, $M_i$ runs in time $t_M(n)$, and
(3) for all $i,j\in\{1,\ldots,m\}$, $S^{M_i,M_j}$ runs in time $t_S(n)$ (here, $t_S(n)$ does not take the running times of $M_i$ and $M_j$ into account).

Then, there exists a $t(n)$-time randomized algorithm that solves $\Pi$ with high probability, where
$t(n)\leq \widetilde{O}(m^2(t_M(n)+t_S(n))\log(1/\delta)\log m)$.
\end{lemma}
\begin{proof}

Let $x$ be an input.
From the assumption, there exists a selector $S$ such that $\Pr_S[S^{A_0,A_1}(x)=\Pi(x)]\geq 1-1/16m$ for any oracles $A_0,A_1$ and any input $x$.
Here, at least one of $A_0$ and $A_1$ solves $(\Pi,\mathcal{D})$ with success probability $\delta$.
Note that the success probability can be assumed to be $1 - 1 / 16 m$ because
one can repeat the computation of $(S, P)$ $O(\log m)$ times.

We present a randomized algorithm $B$ that solves $\Pi$.
For each $i,j\in\{1,\ldots,m\}$, $B$ runs $S^{M_i,M_j}(x)$ and let $c_{ij}$ be its output.
If there exists $i\in\{1,\ldots,m\}$ such that $c_{ij}=c$ for all $j\in\{1,\ldots,m\}$, $B$ outputs $c$.
Repeat this procedure for $O(\log 1/\delta)$ times.
If $B$ outputs nothing during the iteration, $B$ outputs anything.
Since the overall running time of $S^{M_i,M_j}$ is at most $t_S(n)+t_M(n)\polylog(n)$ for every $i,j$,
the algorithm $B$ runs in time $\widetilde{O}(m^2 (t_M(n)+t_S(n)) \log (1 / \delta) \log m )$.

We claim the correctness of the algorithm $B$.
From the assumption, there exists $i\in\{1,\ldots,m\}$
such that $M_i$ solves $(\Pi,\mathcal{D})$ with success probability
$\delta$.
Since we repeat the procedure $O(\log 1/\delta)$ times,
this event holds with high probability.
Under this event,
 by the property of the selector $S$,
we have 
$\Pr_{S}[ S^{M_i, M_j}(x) = \Pi(x)] \ge 1 - 1 / 16 m$
for all $j$.
By the union bound, with probability at least $15 / 16$,
$c_{i, j} =\Pi(x)$ for every $j$.
Similarly, we also have $c_{j, i} = \Pi(x)$ for every $j$ with probability at least $15 / 16$.
These two properties guarantee that the output of $B$ is equal to $\Pi(x)$.
Overall, with probability at least $1 - 3 / 16$, the algorithm $B$ outputs $\Pi(x)$.
\end{proof}

We now present a \emph{completely uniform} direct product theorem
for any problem that admits a $\polylog(n)$-query selector.

\begin{theorem}
[Direct Product Theorem for Any Problem with Selector]
\label{thm:DPTheorem}
	Let $k \in \Nat$, $\epsilon, \delta > 0$ be parameters that satisfy $\epsilon > \exp{(-\Omega(\delta k))}$.
Let $(\Pi, \mathcal{D})$ be a distributional problem.
Suppose there exists a $t_S(n)$-time selector $S$
from $\Pi$
to $(\Pi, \mathcal{D})$ with success probability $\delta$
that calls an oracle at most $\polylog(n)$ times.

Suppose that
there exists a $t(n)$-time heuristic algorithm solving $(\Pi, \mathcal{D})^k$ with success probability~$\epsilon$.
Then,
there exists a $t'(n)$-time algorithm that solves $\Pi$ with high probability.
Here $t'(n) \le \widetilde{O}((t_S(n)+t(n))\log(1/\delta)\log(1/\epsilon)/\epsilon^2)$.
\end{theorem}

\begin{proof}
Let $A$ be a $t(n)$-time heuristic algorithm solving $(\Pi, \mathcal{D})^k$
with success probability $\epsilon$.
By using the algorithm $M$ of \cref{thm:AlmostUniformDirectProduct},
$M^A$ produces a list of oracle algorithms $M_1, \cdots, M_m$
such that $M_i^A$ computes $(\Pi, \mathcal{D})$ with success probability $1 - \delta$ for some $i \in \numset{m}$, where $m = O(1 / \epsilon)$.
Then, we apply \cref{lem:solve_by_selector} using $M_1^A,\ldots,M_m^A$ as the list of algorithms.
Note that, from \cref{thm:AlmostUniformDirectProduct},
each $M_i^A$ runs in time $t(n)$.
Thus, the algorithm solving $\Pi$ of \cref{lem:solve_by_selector}
runs in time $\widetilde{O}((t(n)+t_S(n))\log(1/\delta)\log(1/\epsilon)/\epsilon^2)$.
\end{proof}

Combining the existence of a selector (\cref{cor:SelectorForKabCount}) and 
the direct product theorem (\cref{thm:DPTheorem}),
we obtain the main lower bound result of this paper (\cref{thm_Main}).

\begin{proof}[Proof of \cref{thm:fine-grained_direct_product_theorem_Kab}]
We prove the contrapositive.
Assume that there exists a $t(n)$-time heuristic algorithm that solves
$(\KabCount, \DistKab{a}{b}{n})^k$ with success probability $\epsilon := n^{-\alpha / 4}$,
where $t(n) = n^{a - \alpha}$.

By \cref{cor:SelectorForKabCount},
there exists a $\widetilde{O}(n^2)$-time selector using $\polylog(n)$ queries 
from $\KabCount$ to $(\KabCount, \DistKab{a}{b}{n})$ with success probability $\delta\defeq 1-(\log n)^{-C_H}$,
where $C_H>0$ is a constant depends only on $H$.
We choose $k = O( (\log 1 / \epsilon) / \delta) \le O(\alpha \log n)$ large enough 
so that the assumption of \cref{thm:DPTheorem} is satisfied.
By \cref{thm:DPTheorem},
we obtain a $t'(n)$-time algorithm that solves $\KabCount$,
where $t'(n) = \widetilde{O}((n^2 + t(n)) \cdot n^{\alpha / 2}) \le \widetilde{O}(n^{a - \alpha / 2}) $.
This contradicts \cref{thm:KabCountSETH}.
\end{proof}

\section{Fine-Grained XOR Lemma} \label{sec:parity}
In this section, we show a XOR lemma in the context of
fine-grained complexity.
We focus on the \emph{XOR problem} $\kparity{k}$ defined as follows.
\begin{definition}
Let $\Pi$ be a problem such that $\Pi(x)\in\binset$ for any input $x$.
For a parameter $k\in\Nat$, let $\kparity{k}$ be the problem of computing $\sum_{i=1}^k \Pi(x_i)\pmod 2$ on input $(x_1,\ldots,x_k)$.
\end{definition}

Throughout this section, we consider decision problems unless otherwise noted.
For a distributional problem $(\Pi,\mathcal{D})$, let $\mathcal{D}^k$ be the direct product of $\mathcal{D}$ (see \cref{def:direct_product}).
Suppose that there is a selector from $\Pi$ to $(\Pi,\mathcal{D})$
that makes at most $\polylog(n)$ queries.
The aim of this section is to derive the average-case hardness of the distributional problem $(\kparity{k},\mathcal{D}^k)$ from the worst-case hardness assumption of $\Pi$ (see \cref{thm:XORlemma}).
To this end, we combine Direct Product Theorem (\cref{thm:DPTheorem}) and the well-known list-decoding technique for the Hadamard code due to Goldreich and Levin~\cite{GL89}.
Let us restate the Goldreich-Levin theorem as follows.
\begin{theorem}[Goldreich-Levin Theorem~\cite{GL89}] \label{thm:GLtheorem}
Let $(\Pi,\mathcal{D})$ be a distributional problem and
let $k=k(n)\in\Nat,\epsilon=\epsilon(n)>0$ be parameters.
Then, there exists an algorithm $M$ that, given access to an oracle $A$ solving $(\kparity{k},\mathcal{D}^k)$ with success probability $1/2+\epsilon$, produces with high probability a list of deterministic oracle algorithms $M_1,\ldots,M_m$ such that, for some $t\in\{1,\ldots,m\}$, the oracle algorithm $M_t^A$ solves $(\Pi,\mathcal{D})^{2k}$ with success probability $2/3$.
Here, $m=O(k/\epsilon^2)$.

If an oracle $A$ can be computed in time $T_A(n)$, then each $M_i^A$ runs in time $O\bigl(T_A(n)k^{2.5}/\epsilon^2\bigr)$ for any $i$, and $M$ runs in time $O\bigl(m\cdot T_A(n) k^{2.5}/\epsilon^2\bigr)=O\bigl(T_A(n)k^{3.5}/\epsilon^4\bigr)$.
\end{theorem}
\begin{proof}
Let $(\Pi,\mathcal{D})$ be the distributional problem and
$k=k(n),\epsilon=\epsilon(n)$ be the parameters mentioned in \cref{thm:GLtheorem}.
Consider the following problem $\Pi'$:
Given $2k$ instances $x_1,\dots,x_{2k}$ of $\Pi$ and $r\in\{0,1\}^{2k}$,
compute $\sum_{i=1}^{2k} r_i \cdot \Pi(x_i) \bmod 2$.
Let $(\Pi',\mathcal{D}')$ be a distributional problem, where, in $\mathcal{D}'$,
the input is sampled as $(x_1,\dots,x_{2k})\sim\mathcal{D}^{2k}$ and $r\sim\Unif{\binset^{2k}}$.
Note that, if $r\sim\Unif{\binset^{2k}}$,
with probability at least $\binom{2k}{k}/2^{2 k} \ge 1/(2\sqrt{k})$, the vector $r\in\binset^{2k}$ has exactly $k$ ones.
Here, we used the well-known inequality
$\binom{2k}{k}\geq \bigl(1-\frac{1}{8k}\bigr)\frac{4^k}{\sqrt{\pi k}}$.
Conditioned on this event, the distributional problem $(\Pi',\mathcal{D}')$
is equivalent to $(\kparity{k},\mathcal{D}^k)$.
Let $A'$ be the algorithm that takes $2k$ instances $x_1, \dots, x_{2 k}$ and $r \in \binset^{2 k}$ as input and outputs $A(x_{i_1}, \dots, x_{i_k})$ if $r$ contains exactly $k$ ones in the position of $i_1 < \dots < i_k$; otherwise outputs a random bit.
This algorithm $A'$ runs in time $O(t(n))$ and solves $(\Pi',\mathcal{D}')$ with success probability at least $1/2+\epsilon/(2\sqrt{k})$.
In other words,
\begin{align*}
\Pr_{A', x_1, \ldots, x_{2 k}, r\sim\Unif{\binset^{2k}}}\left[
A'(x_1,\ldots,x_k,r)=\sum_{i=1}^{2k}r_i\Pi(x_i) \bmod 2
\right]
\geq \frac{1}{2}+\frac{\epsilon}{2\sqrt{k}}.
\end{align*}

Now we present the algorithm $M$ mentioned in \cref{thm:XORlemma}.
We say that an input $(x_1,\ldots,x_{2k})$ is \emph{good} if
\begin{align*}
\Pr_{A',r\sim\Unif{\binset^{2k}}}\left[
A'(x_1,\ldots,x_k,r)=\sum_{i=1}^{2k}r_i\Pi(x_i) \bmod 2
\right]
\geq \frac{1}{2}+\frac{\epsilon}{4\sqrt{k}}.
\end{align*}
We claim that at least $\epsilon/(4\sqrt{k})$ fraction of $(x_1,\ldots,x_k)$
are good.
To see this, let $\mathcal{E}$ be the
event that $A'$ success
(i.e., $A'(x_1,\ldots,x_{2k},r)=\sum_{i=1}^{2k}r_i\Pi(x_i) \pmod 2$)
and let $\mathcal{F}$ be the event that
$(x_1,\ldots,x_{2k})$ is good.
Assume $\Pr[\mathcal{F}]<\epsilon/(4\sqrt{k})$.
Then, from the property of $A'$ and the assumption, we have
\begin{align*}
\frac{1}{2}+\frac{\epsilon}{2\sqrt{k}} &\leq
\Pr[\mathcal{E}] \\
&\leq \Pr[\mathcal{E}|\mathcal{F}]\Pr[\mathcal{F}]+\Pr[\mathcal{E}|\text{not $\mathcal{F}$}]\Pr[\text{not $\mathcal{F}$}] \\
&< \frac{\epsilon}{4\sqrt{k}} + \left(\frac{1}{2}+\frac{\epsilon}{4\sqrt{k}}\right) = \frac{1}{2}+\frac{\epsilon}{2\sqrt{k}}.
\end{align*}
Thus we have $\Pr[\mathcal{F}]=\Pr[\text{$(x_1,\ldots,x_{2k})$ is good}]\geq \epsilon/(4\sqrt{k})$.

Let $m=24k/\epsilon^2$
and $\ell$ be the minimum integer satisfying $m\leq 2^\ell$.
The algorithm $M$ produces a list $M_1,\ldots,M_{2^\ell}$ such that,
for some $i\in\{1,\ldots,2^{\ell}\}$,
$M_i^{A'}$ solves $(\Pi,\mathcal{D})^{2k}$ for good inputs.
Let $s^{(1)},\ldots,s^{(\ell)}\sim \Unif{\{0,1\}^{2k}}$ be $\ell$ i.i.d.~random vectors.
Construct $m$ distinct nonempty subsets $T_1,\ldots,T_m\subseteq [\ell]$ in a canonical way and let $r^{(i)}\defeq \sum_{j\in T_i}s^{(i)}_j$.
Note that, for every $i\neq i'$, $r^{(i)}$ and $r^{(i')}$ are pairwise independent random vectors and each $r^{(i)}$ is drawn from $\Unif{\{0,1\}^{2k}}$.

Now we present the list of oracle algorithms $M_1,\ldots,M_{2^{\ell}}$.
For each $t\in\{1,\ldots,2^{\ell}\}$,
the algorithm $M_t^{A'}$ works as follows.
Write $t=\sum_{j=1}^{\ell} 2^{j-1}w_j$ as a binary extension.
In other words, $(w_1,\ldots,w_{\ell})$ can be seen as an
$\ell$-bits of advice.
The bit $w_j$ tells us the value of $\langle \Pi(x),s^{(j)} \rangle \defeq \sum_{i=1}^{2k} \Pi(x_i)s^{(j)}_i\pmod 2$.
Note that, for some $t$, this equality holds for all $j=1,\ldots,\ell$.

Suppose that the input $(x_1,\ldots,x_{2k})$ is good.
Given $(w_1,\ldots,w_\ell)$,
for every $i=1,\ldots,m$, $M_t^{A'}$ does the following:
First, $M_t^{A'}$ computes $W^{(i)}\defeq \sum_{j\in T_i}w_j$.
Note that, for some $t$, we have $W^{(i)} = \sum_{j\in T_i}\langle \Pi(x),s^{(j)} \rangle = \langle \Pi(x), r^{(i)} \rangle$.
Then, for every index $l\in\{1,\ldots,2k\}$, $M_t^{A'}$ calls the oracle and obtain
$A'(x_1,\ldots,x_{2k},r^{(i)}_1,\ldots,r^{(i)}_{l-1},\overline{r^{(i)}_l},r^{(i)}_{l+1},\ldots,r^{(i)}_{2k})$,
where $\overline{z}\defeq 1-z$ for $z\in\binset$.
The output $O^{(i)}$ satisfies $W^{(i)}+O^{(i)}=\Pi(x_l)\pmod 2$ if $A'$ success.
This happens with probability $1/2+\epsilon/(4\sqrt{k})$ since $(x_1,\ldots,x_{2k})$ is good.
We repeat this for $Q=96k^{1.5}/\epsilon^2$ times and then we can compute $\Pi(x_l)$ by taking the majority among the $Q$ trials with successes probability at least $1-\frac{1}{12k}$ for each $l=1,\ldots,2k$.
To see this, let $Z_i$ be a binary indicator random variable such that $Z_i=1$ if and only if $W^{(i)}+O^{(i)}=\Pi(x_l)$.
Let $Z=Z_1+\cdots+Z_Q$.
It suffices to show $\Pr[Z> Q/2]\geq 1-\frac{1}{3k}$.
Note that $\E[Z]\geq \frac{Q}{2}+\frac{\epsilon Q}{4\sqrt{k}}$ and $\Var[Z]= \sum_{i=1}^Q \Var[Z_i]\leq Q$ since the random variables $Z_i$ are pairwise independent.
From the Chebyshev inequality, we obtain
\begin{align*}
\Pr\left[Z\leq \frac{Q}{2}\right] 
&\leq \Pr\left[|Z-\E[Z]|\geq \frac{\epsilon Q}{4\sqrt{k}}\right] \\
&\leq \Pr\left[|Z-\E[Z]|\geq \frac{\epsilon \sqrt{Q}}{4\sqrt{k}}\sqrt{\Var[Z]}\right] \\
&\leq \frac{16\sqrt{k}}{\epsilon^2Q} \leq \frac{1}{6k}.
\end{align*}
Here, recall that the Chebyshev inequality
asserts
\begin{align*}
    \Pr\left[|Z-\E[Z]|\geq \xi\sqrt{\Var[Z]}\right]
    \leq \frac{1}{\xi^2}
\end{align*}
for any $\xi>0$.
Then, from the union bound over $2k$ indices, $M_t^{A'}$ (for the appropriate $t$) computes $(\Pi(x_1),\ldots,\Pi(x_{2k}))$ with probability at least $2/3$.

Note that $M_i^{A'}$ is deterministic without loss of generality since
the coin flips can be given by $M$.
The success probability of $M_i^{A'}$ is at least $(2/3)\cdot (\epsilon/(4\sqrt{k})) \geq \epsilon/(6\sqrt{k})$ since input $(x_1,\ldots,x_{2k})$ is good with probability at least $\epsilon/(4\sqrt{k})$.
The running time of $M_i^{A'}$ is $O(Qk)=O(k^{2.5}/\epsilon^2)$ for all $i\in\{1,\ldots,m\}$.
Thus, if $A'$ is a $T_A(n)$-time algorithm, then we can construct $M_i$ as a deterministic $O(T_A(n) k^{2.5}/\epsilon^2)$-time algorithm.
The total running time of $M$ is at most $m\cdot O(T_A(n) k^{3.5}/\epsilon^4)$
since $M$ constructs $m=O(k/\epsilon^2)$ algorithms each of them runs in time $O(T_A(n) k^{2.5}/\epsilon^2)$.
\end{proof}

Now we prove the main result of this section.

\begin{theorem}[XOR Lemma for Any Problem with Selector] \label{thm:XORlemma}
Let $k\in\Nat,\epsilon,\delta>0$ be parameters satisfying $\epsilon>\exp(-\Omega(\delta k))$.
Let $(\Pi,\mathcal{D})$ be a distributional decision problem.
Suppose there exists a $t_S(n)$-time selector $S$ from $\Pi$ to $(\Pi,\mathcal{D})$ with success probability $\delta$ that calls an oracle at most $\polylog(n)$ times.

Suppose that there exists a $t(n)$-time heuristic algorithm solving $(\kparity{k},\mathcal{D}^k)$ with success probability $\frac{1}{2}+\epsilon$.
Then, there exists a $t'(n)$-time
randomized algorithm that solves $\Pi$
with probability $2/3$.
Here $t'(n)\leq \widetilde{O}\bigl((t_S(n)+t(n))\cdot \log(1/\delta) (k/\epsilon)^8\bigr)$.
\end{theorem}
\begin{proof}
From the assumption, we have a $t(n)$-time heuristic algorithm $A$ solving $(\kparity{k},\mathcal{D}^k)$ with success probability $\epsilon$.
Then, from \cref{thm:GLtheorem} with using $A$ as the oracle, 
we obtain a list of deterministic oracle algorithms $M_1,\ldots,M_m$
such that $M_i^C$ solves $(\Pi,\mathcal{D})^{2k}$ with success probability $\Omega\bigl(\epsilon/\sqrt{k}\bigr)$
for some $i\in\{1,\ldots,m\}$,
where $m=O(k/\epsilon^2)$.
Each of $M_i$ runs in time $O(t(n) k^{2.5}/\epsilon^2)$ if we take the running time of $C$ into account.
This list can be constructed in time $O(t(n) k^{3.5}/\epsilon^4)$.

Let $\delta>0$ be the parameter mentioned in \cref{thm:XORlemma}.
For each $i\in\{1,\ldots,m\}$,
apply \cref{thm:DPTheorem} using $M_i$ as the oracle $C$ mentioned in \cref{thm:DPTheorem}.
This yields a list $M_{i,1},\ldots,M_{i,m'}$ of deterministic algorithms for each $i\in\{1,\ldots,m\}$, where $m'=O(1/\epsilon)$.
Moreover, if $M_{i^*}$ solves $(\Pi,\mathcal{D})^{2k}$, then $M_{i^*,j^*}$ solves
$\Pi$ with success probability $\delta$ for some $j^*\in\{1,\ldots,m'\}$.
For every $i,j$, $M_{i,j}$ runs in time $O(t(n) k^{2.5}/\epsilon^2 \cdot (\log 1/\delta)/\epsilon) \leq O(t(n) k^{2.5}\log(1/\delta)/\epsilon^3)$.

Now we have a list $(M_{i,j})$ of $mm'=O(k/\epsilon^3)$ deterministic algorithms.
From \cref{lem:solve_by_selector},
there exists an algorithm $B$ that solves $\Pi$ with high probability.
The overall running time of $B$ is at most 
$\widetilde{O}\bigl((mm')^2(t_S(n)+t(n)k^{2.5}\log(1/\delta)/\epsilon^3)\log(\sqrt{k}/\epsilon)\log (mm')\bigr) \leq \widetilde{O}\bigl((t_S(n)+t(n))\cdot \log(1/\delta) (k/\epsilon)^8\bigr)$.
\end{proof}

\subsection{Application 1: \texorpdfstring{$\oplus \mathsf{EMB}_{\mathrm{col}}^{(H)}$}{Embcol-Parity}}
Let $H$ be a fixed graph.
Consider the problem
$\numberembcolparity{H}$ of computing the parity of
$\embcol{G}{H}$ for a given graph $G$.
For a parameter $k$, let $\distributionparity{H}$ be the distribution of random graphs that is a direct sum of $k$ i.i.d.~graphs drawn from $\distribution{H}$.
That is, let $G_1,\ldots,G_k\sim \distribution{H}$ be i.i.d.~random graphs.
Suppose that $G(V_i)\cap G(V_j)=\emptyset$ for any $i\neq j$.
Then, the graph $G$ defined by $V(G)=\bigcup_{i=1}^k V(G_i)$ and $E(G)=\bigcup_{i=1}^k E(G_i)$ forms the distribution $\distributionparity{H}$.
Let $\existembcol{H}$ be the decision problem in which we are asked to decide whether $\embcol{H}{G}>0$ or not for a given graph $G$.
This subsection is devoted to the following result.
\begin{theorem} \label{thm:average_hardness_embcolparity}
Suppose that there exists a $t(n)$-time heuristic algorithm solving the distributional problem $(\numberembcolparity{H},\distributionparity{H})$ with success probability $\frac{1}{2}+\epsilon$ far any $k=O(\log \epsilon^{-1})$.
Then, there exists a $t(n)\cdot (\log n/\epsilon)^{O(1)}$-time randomized
algorithm solving $\existembcol{H}$
with probability $2/3$.
\end{theorem}

The proof of \cref{thm:average_hardness_embcolparity} consists of the following three steps.
First, we present a randomized reduction of $\existembcol{H}$ to $\numberembcolparity{H}$ in the worst-case sense.
Then, we check that the parity problem $\numberembcolparity{H}$ admits a $\widetilde{O}(n^2)$-time selector with $\polylog(n)$ queries.
Finally, we apply \cref{thm:XORlemma} to boost the error tolerance.
The second and third steps imply \cref{thm:XOR_for_embcolparity}.
More specifically, we obtain the following.

\begin{theorem}[Refinement of \cref{thm:XOR_for_embcolparity}] \label{thm:XOR_for_embcolparity_refine}
Let $H$ be an arbitrary graph.
Suppose that there exists a $T(n)$-time randomized heuristic algorithm that solves
$(\numberembcolparity{H}, \distributionparity{H})$ with success probability greater than $\frac{1}{2}+\epsilon$ for any $k=O(\log \epsilon^{-1})$.
Then, there exists a $T(n)(\log n/\epsilon)^{O(1)}$-time randomized algorithm that solves $\numberembcolparity{H}$ with probability at least $2/3$ for any input.
\end{theorem}

\begin{remark}
\Cref{thm:XOR_for_embcolparity}
immediately follows from \cref{thm:XOR_for_embcolparity_refine}
(substitute $\epsilon=n^{-c}$ to \cref{thm:XOR_for_embcolparity_refine}).
\end{remark}

\paragraph*{Parity vs.~Detection.}
\begin{lemma} \label{lem:detection_to_parity}
Suppose that there exists a $t(n)$-time randomized algorithm solving $\numberembcolparity{H}$ for any input with probability at least $2/3$.
Then, there exists a $t'(n)$-time randomized algorithm that solves $\existembcol{H}$ with probability at least $2/3$.
Here, $t'(n)=O(2^{|E(H)|}t(n))$.
\end{lemma}
\begin{proof}
The proof is essentially given in Appendix A of \cite{BBB19}.
For completeness, we present the proof.
Consider the polynomial $P_G:\mathbb{F}_2^{E(G)}\to\mathbb{F}_2$ defined as
\begin{align*}
    P_G(x)\defeq \sum_{\substack{F\subseteq E(G):\\F\text{ is isomorphic to }H}} \prod_{e\in F} x_e.
\end{align*}
Then, $G$ does not contain $H$ if and only if $P_G(\cdot)\equiv 0$.
The degree of $P_G$ is $|E(H)|$.
Moreover, if $P_G(\cdot)\not\equiv 0$, then $P_G(z)=1$ for at least $2^{-|E(H)|}$ fraction of $z\in\mathbb{F}_2^{|E(G)|}$ (see, e.g., Lemma 2.6 of \cite{NS94}).

Now we present an algorithm that solves $\existembcol{H}$ using an oracle that solves $\numberembcolparity{H}$.
Let $m=100\cdot 2^{|E(H)|}$ and sample $m$ i.i.d.~random vectors $z_1,\ldots,z_m \sim \Unif{\mathbb{F}_2^{E(G)}}$.
Then, compute $P_G(z_1),\ldots,P_G(z_m)$.
If $P_G(z_i)=1$ for some $i$, output YES.
Otherwise, output NO.
Note that one can compute $P_G(\cdot)$ by solving $\numberembcolparity{H}$
since $P_G(\cdot)$ is a polynomial over $\mathbb{F}_2$.

If $G$ does not contain $H$, the algorithm outputs NO with probability $1$.
If $G$ contains $H$, the probability that the algorithm outputs NO is at most $(1-2^{-|E(H)|})^m \leq \mathrm{e}^{-100}$.
\end{proof}

\paragraph*{Selector for $\oplus \mathsf{EMB}^{(H)}_{\mathrm{col}}$.}
\begin{theorem}
\label{thm:SelectorForEmbcolHparity}
There exists a selector $S$
from $\numberembcolparity{H}$ to $(\numberembcolparity{H}, \distribution{H})$ with success probability $1 - 1 / \polylog(n)$ such that
(1) $S$ runs in time $\widetilde{O}(n^2)$, and
(2) The number of oracle accesses is at most $\polylog(n)$.
\end{theorem}
\begin{proof}
The proof is basically the same as that of \cref{thm:SelectorForEmbcolH}.
To construct a worst-case-to-average-case reduction $R'$ for $\numberembcolparity{H}$,
we encode $\numberembcolparity{H}$ to the low-degree polynomial $\embcolpoly{n}{H}{\mathbb{F}_{2^t}}$.
Note that, since $\mathbb{F}_{2^t}$ has characteristic $2$ (i.e., $a+a=0$ for any $a\in\mathbb{F}_{2^t}$), computing the polynomial $\embcolpoly{n}{H}{\mathbb{F}_{2^t}}(x)$ 
for $x\in \binset^{E(H)\times K_n}$ is equivalent
to solving $\numberembcolparity{H}$ by
regarding the input $x$ as the edge indicator of a graph.
Using \cref{cor:worsttoaverage},
we reduce computing $\embcolpoly{n}{H}{\mathbb{F}_{2^t}}$
to solving the distributional problem
$(\embcolpoly{n}{H}{\mathbb{F}_{2^t}}, \distU{n}{H}{\mathbb{F}_{2^t}})$.
Moreover, we can reduce
$(\embcolpoly{n}{H}{\mathbb{F}_{2^t}}, \distU{n}{H}{\mathbb{F}_{2^t}})$
to
$(\embcolpoly{n}{H}{\mathbb{F}_{2^t}}, \distU{n}{H}{\mathbb{F}_2})$
with query complexity $(\log n)^{O(|E(H)|)}$
using the technique of~\cite{BBB19}.
This yields a worst-case-to-average-case reduction for $\numberembcolparity{H}$ (c.f., \cref{thm:worst-to-average_precise}).

Similarly, a slight modification of the interactive proof system $\IP$ of \cref{thm:embcolIP} works for $\numberembcolparity{H}$.
To be more specifically, let us consider an interactive proof system $\IP'$ for the statement ``$\numberembcolparity{H}(G)=b$".
The protocol $\IP'$ is the same as $\IP$ except for using $\mathbb{F}_{2^t}$ instead of $\Fq$.
Note that the equation~\cref{eq:embcolpolydownwardreducibility} holds even for $\embcolpoly{n}{H}{\mathbb{F}_{2^t}}$.
Moreover, computing the polynomial $\embcolpoly{n}{H}{\mathbb{F}_{2^t}}$ can be reduced to computing $\embcolpoly{n}{H}{\mathbb{F}_2}$ using the aforementioned technique of~\cite{BBB19}.

Using the interactive proof system $\IP'$ for $\numberembcolparity{H}$, we can construct an $\widetilde{O}(n^2)$-time $\polylog(n)$-query instance checker $C'$ for $\numberembcolparity{H}$ (see \cref{thm:InstanceCheckerForEmbcol}).
Combining the instance checker $C'$ and the worst-case-to-average-case reduction $R'$,
we can construct the desired selector (see the proof of \cref{thm:SelectorForEmbcolH}).
\end{proof}

\paragraph*{XOR lemma for $\oplus\mathsf{EMB}_{\mathrm{col}}^{(H)}$ (proof of \cref{thm:XOR_for_embcolparity_refine}).}
Assume that there exists a $t(n)$-time heuristic algorithm $A$
that solves $(\numberembcolparity{H},\distributionparity{H})$
with success probability $\epsilon$.
Note that the distributional problem $(\bigoplus (\numberembcolparity{H})^k, (\distribution{H})^k)$
is equivalent to the distributional problem
$(\numberembcolparity{H},\distributionparity{H})$.
Hence, the algorithm $A$ also solves $(\bigoplus (\numberembcolparity{H})^k,(\distribution{H})^k)$.
From \cref{thm:SelectorForEmbcolHparity},
there exists an $\widetilde{O}(n^2)$-time selector using $\polylog(n)$ oracle
accesses from $\numberembcolparity{H}$ to $(\numberembcolparity{H},\distribution{H})$
with success probability $\delta\defeq 1-(\log n)^{-C}$ for some constant $C>0$ that depends only on $H$.
Let $k=k(n)$ be the parameter such that the assumptions of \cref{thm:XORlemma} is satisfied.
Since $\delta =1-(\log n)^{-C}$, we can set $k=O(\log \epsilon^{-1})$.
Then, by \cref{thm:XORlemma}, we have an $t'(n)$-time
ranndomized algorithm that solves $\numberembcolparity{H}$ with high probability,
where $t'(n)=\widetilde{O}((n^2+t(n))\cdot (k/\epsilon)^8=t(n)\cdot (\log n/\epsilon)^{O(1)}$ (here we assume $t(n)\geq n^2$).

\paragraph*{Proof of \cref{thm:average_hardness_embcolparity}.}
We combine \cref{lem:detection_to_parity,thm:XOR_for_embcolparity_refine}.
Suppose that there exists a $t(n)$-time heuristic algorithm solving $(\numberembcolparity{H},\distributionparity{H})$ with success probability $\epsilon$.
From \cref{thm:XOR_for_embcolparity_refine},
there exists a $t(n)\cdot (\log n/\epsilon)^{O(1)}$-time randomized algorithm for $\numberembcolparity{H}$.
Then, from \cref{lem:detection_to_parity}, we obtain
a $2^{|E(H)|}\cdot t(n)\cdot (\log n/\epsilon)^{O(1)}$-time randomized algorithm for $\existembcol{H}$.

\subsection{Application 2: \texorpdfstring{$\KaParity$}{\#Ka-Parity}}
\label{subsec:KaParity}
Recall that $\KaParity$ is the problem of computing the parity of the number of $K_a$ subgraphs contained in a given graph.
This subsection is devoted to prove \cref{thm:worst_to_average_KaParity-intro}.
Recall that
$\DistER$ is the distribution of the Erd\H{o}s-R\'enyi graph $G(n,1/2)$,
and $\DistERDisjointUnion$ is the distribution of the disjoint union of $k$ random graphs $G_1,\ldots,G_k$ each of which is independently drawn from $\DistER$.
\begin{theorem}[Refinement of \cref{thm:worst_to_average_KaParity-intro}]
\label{thm:worst_to_average_KaParity_refine}
Suppose that there exists a $T(n)$-time randomized heuristic algorithm that solves $(\KaParity,\DistERDisjointUnion)$
with success probability $\frac{1}{2}+\varepsilon$
for any $k=O(\log \varepsilon^{-1})$.
Then, there exists a $T(n) (\log n/\varepsilon)^{O(1)}$-time randomized algorithm
that solves $\KaParity$ for any input with probability $2/3$.
\end{theorem}
\begin{proof}[Proof of \cref{thm:worst_to_average_KaParity-intro}]
\Cref{thm:worst_to_average_KaParity_refine}
directly implies \cref{thm:worst_to_average_KaParity-intro} (let $\varepsilon=n^{-\epsilon}$).
\end{proof}
The core of the proof of \cref{thm:worst_to_average_KaParity_refine} is the existence of the following efficient selector.
\begin{lemma}
\label{lem:SelectorForKaParity}
There exists an $\widetilde{O}(n^2)$-time selector $S$
from $\KaParity$ to the distributional problem $(\KaParity, \DistER)$ with success probability $1 - 1 / \polylog(n)$.
Moreover, the number of oracle accesses of $S$ is
at most $\polylog(n)$.
\end{lemma}
\begin{proof}
The proof of \cref{lem:SelectorForKaParity}
is similar to that of \cref{cor:SelectorForKabCount}.
From \cref{thm:SelectorForEmbcolHparity},
we have obtain a selector from $\numberembcolparity{K_a}$ to $(\numberembcolparity{K_a},\distribution{K_a})$.
Then, we use the reduction by Boix-Adser\'a, Brennan, and Bresler~\cite{BBB19}.
They reduced $(\numberembcolparity{K_a},\distribution{K_a})$
to
$(\KaParity,\DistER)$
with preserving the success probability up to a constant factor (Lemma 3.10 of \cite{BBB19}).
Using their reduction, each query of the selector $S$ can be replaced by the reduction.
This yields a selector from $\numberembcolparity{K_a}$ to $(\KaParity,\DistER)$.
We then use the reduction
of Lemma 3.3 of Boix-Adser\'a, Brennan, and Bresler~\cite{BBB19}.
They reduced $\KabParity$ to $\numberembcolparity{K_a}$.
Specifically, if $\numberembcolparity{K_a}$
can be solved in time $t(n)$, then there exists a $t(n)+O(n^2)$-time algorithm for $\KaParity$.
\end{proof}

\begin{proof}[Proof of \cref{thm:worst_to_average_KaParity_refine}]
Suppose that there exists a $T(n)$-time randomized heuristic algorithm that solves $(\KaParity,\DistERDisjointUnion)$
with success probability $\frac{1}{2}+\epsilon$
for any $k=O(\log \epsilon^{-1})$.
Note that $(\KaParity,\DistERDisjointUnion)$ is equivalent to $((\KabParity)^{\oplus k},(\DistER)^k)$.
From \cref{thm:XORlemma,lem:SelectorForKaParity},
we obtain an $t'(n)$-time randomized algorithm that solves $\KabParity$ with probability $2/3$, where $t'(n)=(n^2+t(n)(\log n/\epsilon)^{O(1)}=t(n)\cdot (\log n/\epsilon)^{O(1)}$ (here, we assume $t(n)=\Omega(n^2)$ and let $\delta=1-(\log n)^{-C})$ and $k=O(\log \epsilon^{-1})$).

\end{proof}

\section*{Acknowledgement}
We thank Marhsall Ball for helpful discussion on Proof of Work.
Shuichi Hirahara is supported by ACT-I, JST.
Nobutaka Shimizu is supported by JSPS KAKENHI Grant Number JP19J12876.

\bibliographystyle{alphaabbr}
\bibliography{ref}
\end{document}